\documentclass[11pt]{amsart}
\usepackage{amssymb}

\usepackage{amsfonts}
\usepackage{xcolor}
\usepackage{graphicx}
\usepackage{float}
\usepackage{lmodern}
\usepackage{graphicx}
\usepackage{cite}
\usepackage[bottom]{footmisc}
\theoremstyle{plain}

\newtheorem{corollary}{Corollary}[subsection]

\newtheorem{definition}{Definition}[subsection]

\newtheorem{lemma}{Lemma}

\newtheorem{proposition}{Proposition}[subsection]
\newtheorem{remark}{Remark}[section]

\newtheorem{theorem}{Theorem}[subsection]

\numberwithin{equation}{section}

\begin{document}
\title{Nonlinear coherent states associated with a measure on the positive real half line}
\maketitle
\begin{center}
\author{S. Twareque Ali$^{*}$\footnote{
This paper is a development of project that Z. Mouayn during his visit to
Concordia university on September 2014 has started with Professor S.
Twareque Ali. Later, on January 2016, Professor S. Twareque Ali passed away.
This work is dedicated to his memory.}, Zouha\"{\i}r Mouayn$^{\flat}$ \ and
\ Khalid Ahbli$^{\Upsilon}$ }

\begin{scriptsize}
$^{*}$ Department of Mathematics and Statistics, Concordia University,\vspace{-0.2em}\\
Montr\'eal, Canada \\
${}^{\flat}$ Department of Mathematics, Faculty of
Sciences and Technics (M'Ghila),\vspace{-0.2em}\\ P.O. Box. 523, B\'{e}ni Mellal, Morocco.\\
$^{\Upsilon}$ Department of Mathematics, Faculty of
Sciences, Ibn Zohr University,\vspace{-0.2em}\\ P.O. Box. 8106, Agadir, Morocco.\vspace*{0.2mm}\vspace*{0.2mm}
\end{scriptsize}
\end{center}
\begin{abstract}
We construct a class of generalized nonlinear coherent states by means of a newly obtained class of 2D complex orthogonal polynomials. The associated coherent states transform is discussed. A polynomials realization of the basis of the quantum states Hilbert space is also obtained. Here, the entire structure owes its existence to a certain measure on the positive real half line, of finite total mass, together with all its moments. We illustrate this construction with the example of the measure $r^\beta e^{-r}dr$, which leads to a new generalization of the true-polyanalytic Bargmann transform.\\
\end{abstract}

\begin{scriptsize}
KEYWORDS: Nonlinear coherent states; 2D complex orthogonal polynomials;  Bargmann-type transform, positive measure on $\mathbb{R}_+$.
\\ 
AMS CLASSIFICATION: 33C45; 81R30; 35A22; 
\end{scriptsize}
\section{Introduction}
Nonlinear coherent states (NLCS) were first introduced explicitly in \cite%
{Matos} and \cite{Manko}, although having appeared implicitly in \cite%
{Shanta} in a compact form. These states have attracted much attention in
recent decades, mostly because they exhibit nonclassical properties. As
mentioned in \cite{RT}, up to now many quantum optical states such as $q$%
-deformed coherent states \cite{Manko}, photon added coherent states \cite%
{Siv1999,Siv2000,Naderi}, the center of mass motion of a
trapped ion \cite{Matos}, some nonlinear phenomena such as a hypothetical
"frequency blue shift" in high intensity photon beams \cite{Monko1995} and
the binomial state \cite{RT2} have been considered as some kind of NLCS.

Following \cite{RT3} there are different but equivalent ways to introduce
NLCS. We shall adopt the following one. We recall that the \textit{canonical} coherent states \cite{AAG} are written in terms of the so-called Fock basis $\{\varphi_n\}_{n=0}^{\infty}$(or number states): 
\begin{equation}
\vartheta_z =\left(e^{z\overline{z}}\right)^{-\frac{1}{2}}\sum\limits_{n=0}^{+\infty }\frac{\overline{z}^{n}}{\sqrt{n!}}
\varphi_n ,  \label{CCS}
\end{equation}
for each fixed $z\in \mathbb{C}$ where $e^{z\overline{z}}$ is chosen so that to ensure the normalization condition $\langle \vartheta_z|\vartheta_z\rangle =1$. The basis vectors $\{\varphi_n\} $ are orthonormal in the underlying quantum states Hilbert space $\mathcal{H}$, often
termed a Fock space.

 In this paper, we construct a class of generalized nonlinear coherent states (GNLCS) by replacing the monomials $\bar{z}^n$ in $(\ref{CCS})$ by a newly obtained class of 2D complex orthogonal polynomials. The new polynomials, here denoted $P_{n,m}(z,\bar{z},\beta ),\; n,m\geq 0$, were constructed in \cite{MZ} starting from a positive measure $d\mu_\beta$ on the positive real half line with finite total mass. Here, we will fix $m\in\mathbb{Z}_+$ and from the moments of this measure we build a sequence of real numbers $(x_{n,m}^\beta)_{n\geq 0}$ which will be used to define a generalized factorial $x_{n,m}^\beta !$. The later one replaces the factorial $n!$ in the denominator of $(\ref{CCS})$ and enables us to introduce a new class of GNLCS which we denote by $\mu_\beta$-GNLCS for brevity. We discuss the resolution of the identity satisfied by these coherent states and we give a polynomial realization of the basis $\{\varphi_n \}$ by using the shift operators method \cite{AI} which again is based on the sequence $(x_{n,0}^\beta )_{n\geq 0}$. We also write down the associated coherent states transform. It turns out that the range of this transform can be realized as an $L^2$-eigenspace of a Hamiltonian-type operator which is the anti-commutator of two operators defined by their actions on the above polynomials using the sequence $x_{n,m}^\beta $. Here, the remarkable fact is that the entire structure,
consisting of the biorthogonal polynomials, $\mu_\beta$-GNLCS and their corresponding coherent states transform, is completely determined by the choice of a single
measure on the positive real half line. We illustrate this construction with the example of the measure $d\mu_{\beta}(r)=r^\beta e^{-r}dr$, which leads to a new generalization of the true-polyanalytic Bargmann transform \cite{MouaynMN,Abreusampling,abfei14}.\\

The paper is organized as follows. In section 2, we briefly recall the formalism of NLCS and we introduce their construction starting from a positive measure on the real half line. The associated Hamiltonian operator is also discussed and the method for a polynomials realization of the basis $\{\varphi_n \}$ is summarized. In section 3, we review the construction of a general class of 2D orthogonal polynomials. These polynomials are then used to define our $\mu_\beta$-GNLCS whose corresponding coherent states transform together with its range are also discussed. Section 4 deals with an example of the measure $d\mu_{\beta}$ which illustrate our method.
\section{Nonlinear coherent states and their Hamiltonian}

\subsection{Nonlinear coherent states.} The so-called \textit{deformed coherent states} \cite{AAG,Manko} also known as NLCS in the quantum optical literature \cite{VW} are then defined by replacing the factorial $n!$ in the denominator following
the summation sign in $\left( \ref{CCS}\right) $ by $x_{n}!:=x_{1}x_{2}...x_{n},\, x_{0}=0,$ where $\{x_{n}\}_{n=0}^{\infty }$ is an infinite sequence of
positive numbers and, by convention, $x_{0}!=1$. Thus, for each $z\in 
\mathcal{D}$ some complex domain, one defines a generalized version of $\left(
\ref{CCS}\right) $ as 
\begin{equation}
\vartheta_z =(\mathcal{N}(z\bar{z}))^{-1/2}\sum\limits_{n=0}^{+\infty }%
\frac{\bar{z}^{n}}{\sqrt{x_{n}!}}\varphi_n ,\quad z\in \mathcal{D}  \label{NLCS x_m}
\end{equation}%
where again 
\begin{equation}
\mathcal{N}(z\bar{z})=\sum\limits_{n=0}^{+\infty }\frac{\left( z%
\overline{z}\right) ^{n}}{x_{n}!}  \label{normalization}
\end{equation}%
is an appropriate normalizing constant.  It is clear that the vectors $\vartheta_z$ are well defined for all $z$ for which the sum $(\ref{normalization})$
converges, i.e. $\mathcal{D}=\{z\in \mathbb{C},|z|<R\}$ where $R^{2}=\lim_{n\rightarrow +\infty }x_{n}$, with $R>0$ could be
finite or infinite, but not zero. As usual, we
require that there exists a measure $d\eta $ on $\mathcal{D}$ for which the
resolution of the identity condition, 
\begin{equation}
\int_{\mathcal{D}}|\vartheta_z\rangle \langle \vartheta_z|\mathcal{N}(z\bar{z})d\eta (z,%
\bar{z})=\textbf{1}_{\mathcal{H}} \label{reso-ident}
\end{equation}%
holds. Here, $|\vartheta_z\rangle \langle \vartheta_z|\equiv T_z$ means the rank one operator $T_z: \mathcal{H}\longrightarrow \mathcal{H}$ defined by $T_z[\psi]=\langle \vartheta_z | \psi \rangle \vartheta_z,\ \psi\in \mathcal{H}$. In order for $(\ref{reso-ident})$ to be satisfied, the
measure $d\eta $ has to have the form 
\begin{equation}
d\eta (z,\bar{z})=\frac{d\theta }{2\pi }d\lambda (r ),\quad z=r
e^{i\theta }  
\end{equation}%
where the measure $d\lambda $ is a solution of the moment problem 
\begin{equation}
\int_{0}^{R}r ^{2n}d\lambda (r )=x_{n}!,\quad n=0,1,2,...,
\label{Prob-moment}
\end{equation}%
provided that such a solution exists. In most of the cases that occur in
practice, the support of the measure $d\eta $ is the whole domain $\mathcal{D%
}$, meaning that $d\lambda $ is supported on $(0,R)$.

\subsection{NLCS associated with a measure}
We start from a family of measures of the form
\begin{equation}
d\mu_\beta (r):=r^{\beta}d\mu(r)\label{measure}
\end{equation}
where $d\mu(r)$ is a positive measure, does not depend on $\beta $, supported by $(0,L)$, where $L$ could be infinite. Assume that this measure have finite moments for all order and denote $\mathcal{D}_L=\{\xi\in\mathbb{C}, |\xi|<L\}$. Set
\begin{equation}
\mu _{\beta }:=\int_{0}^{L}d\mu_\beta \left( r \right) ,\text{ }\beta \geq
0. \label{measure mu_beta}
\end{equation}
From these moments we consider the following sequence of numbers 
\begin{equation}
x_{n}^{\beta }=\frac{\mu _{n+\beta }}{\mu _{n+\beta -1}}, \text{ \ \ \ }%
x_{n}^{\beta }!:=x_{n}^{\beta }x_{n-1}^{\beta }\ldots x_{1}^{\beta }=\frac{%
\mu _{n+\beta }}{\mu _{\beta }},\text{ \ \ }x_{0}^{\beta }!\equiv 1\text{,}\quad n=1,2,3,..., \label{sequ_x_n_b}
\end{equation}
which allows us to define the following nonlinear coherent states. 
\begin{definition}
For $\beta \geq 0$, the vectors $\vartheta_{z,\beta}\equiv |z;\beta \rangle\in \mathcal{H}$ are defined through the superposition
\begin{equation}
\vartheta_{z,\beta}:=\left( \mathcal{N}_{\beta}(z\overline{z} )\mu _{\beta }\right)^{-\frac{1}{2}}\ \sum_{n=0}^{\infty
}\ \frac{\overline{z}^{n}}{\sqrt{x_n^{\beta}!}}\varphi_n\text{,} \label{NLCS_x_m,beta}
\end{equation}
for each $z$ in $\mathcal{D}_{\beta}=\{\xi\in \mathbb{C},|\xi |<R_\beta=(\lim_{n\rightarrow +\infty }x_{n}^{\beta})^{\frac{1}{2}} \}$. For brevity, these states will be denoted $\mu_\beta -$NLCS.
\end{definition}
The vectors $\{\varphi_n\}_{n=0}^{\infty }$ form an orthonormal basis in the Hilbert space $\mathcal{H}$ and the normalization constant is given by
\begin{equation}
\mathcal{N}_{\beta}(z\bar{z})=\sum\limits_{n=0}^{+\infty }\frac{\left( z%
\overline{z}\right) ^{n}}{x_{n}^{\beta}!} \label{Norm_x_m,beta}
\end{equation}
which converges for each $z\in \mathcal{D}_{\beta}$. 
\begin{proposition}
Let $d\mu_\beta$ be a measure given in $(\ref{measure})$ and assuming that $\mathcal{D}_L \subseteq \mathcal{D}_{\beta} $, then the $\mu_\beta $-NLCS in $(\ref{NLCS_x_m,beta})$ satisfy the resolution of the identity operator of $\mathcal{H}$ as
\begin{equation}
\int_{\mathcal{D}_L}|\vartheta_{z,\beta }\rangle \langle
\vartheta_{z,\beta } |\ d\eta_\beta (z,%
\bar{z})=\mathbf{1}_{
\mathcal{H}}
\end{equation}
where $d\eta_\beta (z,%
\bar{z})=(2\pi)^{-1}d\theta  \mathcal{N}_{\beta}(z\overline{z} )d\mu_\beta (z\bar{z} )$, $z\in\mathcal{D}_L$.
\end{proposition}
\begin{proof}
Let us assume that the measure takes the form $d\eta_\beta (z,%
\bar{z})=\mathcal{N}_{\beta}(z\bar{z})h(z\bar{z})d\nu(z),$
where $h$ is an auxiliary density function to be determined and $d\nu(z)$ is the Lebesgue measure on $\mathbb{C}$. In terms of
polar coordinates $z=\rho e^{i\theta} ,\ \rho>0$ and $\theta\in [0,2\pi)$, the measure can be rewritten as $d\eta_\beta (z,\bar{z})=\pi^{-1} \mathcal{N}_{\beta}(\rho^2 )h(\rho^2)\rho d\rho d\theta.$ Using the expression $(\ref{NLCS_x_m,beta})$ of coherent states, the operator $\mathcal{O}_{\beta}=\int_{\mathcal{D}_L}\left| \vartheta_{z,\beta }
\right\rangle \left\langle\vartheta_{z,\beta }\right| d\eta_\beta (z,%
\bar{z})$ reads successively,
\begin{eqnarray}
\mathcal{O}_{\beta}&=&(\mu_\beta)^{-1}\sum\limits_{n,m=0}^{+\infty} \frac{1}{\sqrt{x_n^\beta !}\sqrt{x_m^\beta !}}\left( 2 \int_0^{L^2}\rho^{n+m}h(\rho^2)\rho d\rho \int_0^{2\pi} e^{i(n-m)\theta}\frac{d\theta}{2\pi} \right)  \vert \varphi_m\rangle \langle \varphi_n\vert\\
&=&(\mu_\beta)^{-1}\sum\limits_{n=0}^{+\infty} \frac{1}{x_n^\beta !}\left(  2\int_0^{L^2}\rho^{2n}h(\rho^2)\rho d\rho \right)  \vert \varphi_n\rangle \langle \varphi_n\vert .
\end{eqnarray}
By a change of variables, we get
\begin{eqnarray}
\mathcal{O}_{\beta}&=&\sum\limits_{n=0}^{+\infty}\frac{1}{\mu_{n+\beta}}\left(  \int_0^{L}r^{n}h(r)dr \right)  \vert \varphi_n\rangle \langle \varphi_n\vert.\label{int_rho}
\end{eqnarray}
Thus, in order to recover the discrete resolution of the identity of $\mathcal{H}$, we need to find $r\mapsto h(r)$ such that
\begin{eqnarray}
\int_0^{L}r^{n}h(r)dr =\mu_{n+\beta}.
\end{eqnarray}
Equations $(\ref{measure mu_beta})-(\ref{sequ_x_n_b})$ suggest us to choose $h(r)=d\mu_{\beta}(r)/dr$ i.e., the Radon-Nikodym derivative of $d\mu_{\beta}$ with respect to $dr$.
Then, equation $(\ref{int_rho})$ reduces to $\mathcal{O}_{\beta}=\sum_{n=0}^{+\infty}\left|\varphi_n\right\rangle\left\langle \varphi_n\right|=\textbf{1}_{\mathcal{H}},$ since $\{|\varphi_n\rangle\}$ is an orthonormal basis of $\mathcal{H}$.
\end{proof}
\subsection{The Hamiltonian of the NLCS }
Man'ko \textit{et al}'s \cite{Manko} approach is based on the two following postulates: $(i)$ The standard annihilation and creation operators are deformed with an intensity dependent function $f(\widehat{n})$ (which is an operator valued function and $f$ can be chosen real and non-negative, i.e. $f^\dag (\widehat{n})=f(\widehat{n})$) as
\begin{eqnarray}
A&=&af(\widehat{n})\label{A}\\
A^\dag &=&f^\dag(\widehat{n})a^\dag\label{A^dag}\\
\left[A,A^\dag \right]&=&(\widehat{n}+1)f(\widehat{n}+1)f^\dag(\widehat{n}+1)-\widehat{n}f^\dag(\widehat{n})f(\widehat{n})
\end{eqnarray}
where $a$, $a^\dag$ and $\widehat{n}=a^\dag a$ are bosonic annihilation, creation and number operators, respectively. $(ii)$ The Hamiltonian of the deformed oscillator in analogy with the harmonic oscillator is found to be 
\begin{equation}
\widehat{H}=\frac{1}{2}\left(AA^\dag +A^\dag A\right)\label{Manko oper}
\end{equation}
and the Hamiltonian of the NLCS is defined as the \textit{normal ordered} or Wick operator of $\widehat{H}$ as
\begin{equation}
H_{\text{NLCS}}=:\widehat{H}:=A^\dag A
\end{equation}
which by equations $(\ref{A})-(\ref{A^dag})$ reads
\begin{equation}
H_{\text{NLCS}}=\widehat{n}[f(\widehat{n})]^2.\label{NLCS oper}
\end{equation}
So that the single mode NLCS obtained as an eigenstate of the annihilation operator can be expressed as
\begin{equation}
\vartheta_{z,f}=\left(\mathcal{N}_f\left(z\overline{z}\right)\right)^{-1/2}\sum\limits_{n=0}^{\infty}\frac{z^n}{\sqrt{x_n!}} \varphi_n\label{NL}
\end{equation}
in terms of the function $f$ by setting $x_n:=nf^\dag (n)f(n)$. So that $|z|\leqslant \lim\limits_{n\rightarrow +\infty}n[f(n)]^2$ ensures that $\mathcal{N}_f\left(z\overline{z}\right)<+\infty$. In our setting, the sequence $x_n$ is taken as $x_n^{\beta}$ therefore the connection between $(\ref{NLCS_x_m,beta})$ and $(\ref{NL})$ is given by
\begin{equation}
x_n^{\beta}=n[f(n)]^2.
\end{equation}
In view of $(\ref{NLCS oper})$, one can see that the Hamiltonian of the $\mu_\beta -$NLCS can be defined through the given sequence $x_n^{\beta}$ as
\begin{equation}
H_{\mu_\beta}:=x_{\widehat{n}}^{\beta}.
\end{equation}
\subsection{A polynomials realization of the basis $\{\varphi_n\}_{n=0}^\infty$}
The orthonormal basis vectors $\{\varphi_n\}_{n=0}^\infty$, that have been used to define the coherent states (\ref{NLCS_x_m,beta}), were taken from some abstract Hilbert space $\mathcal{H}$. On this Hilbert space we may define the operators $a_\beta $ and $a_{\beta}^\dag$ by
\begin{equation}
  a_\beta\varphi_n = \sqrt{x^\beta_{n}}\varphi_{n-1},\;\; a_\beta\varphi_0 =0, \qquad
  a_{\beta}^{\dag}\varphi_n = \sqrt{x^\beta_{n+1}}\varphi_{n+1}.
\label{abst-ops}
\end{equation}
Using these operators, it is now possible to identify the basis vectors $\{\varphi_n\}_{n=0}^\infty$  with another family of real orthogonal polynomials by following \cite{AI}. In order to do this, we make the assumption that the sequence of real numbers, $x_n^{\beta}, \;\; n =1,2, \ldots $, defined in (\ref{sequ_x_n_b}) satisfy the condition,
\begin{equation}
  \sum_{n=1}^\infty \frac 1{\sqrt{x^\beta_n}} = \infty.
\label{div-cond}
\end{equation}
We define the operators,
\begin{equation}
  Q_\beta = \frac 1{\sqrt{2}}\; [a_\beta + a_{\beta}^{\dagger} ]\; , \qquad
  P_\beta = \frac 1{i\sqrt{2}}\; [a_\beta - a_{\beta}^{\dagger} ]\; ,
\label{eq:pos-mom-op}
\end{equation}
which are analogues of the standard position and momentum operators. The operator $Q_\beta$ acts on the basis vectors $\varphi_n$ as
\begin{equation}
  Q_\beta \varphi_n = \sqrt{\frac{x^\beta_n}2}\; \varphi_{n-1} + \sqrt{\frac{x^\beta_{n+1}}2}\; \varphi_{n+1}\; .
\label{eq:pos-op-act}
\end{equation}
If now the sum $\sum_{n=0}^\infty \dfrac 1{\sqrt{x^\beta_n }}$ diverges, the operator $Q_\beta$ is essentially self-adjoint and hence has a unique self-adjoint extension, which we again denote by $Q_\beta$. Let $E^\beta_x , \; x\in \mathbb R$, be the spectral family of $Q_\beta$, so that,
$$ Q_\beta = \int_{-\infty}^\infty x \; dE^\beta_x \; .$$
Thus there is a measure $d\omega_{\beta}(x)$ on $\mathbb R$ such that on the Hilbert space $L^2 (\mathbb R , d\omega_{\beta}(x))$,
$Q^\beta$ is just the operator of multiplication by $x$. Consequently, on this space, the relation (\ref{eq:pos-op-act})
takes the form
\begin{equation}
  x\varphi_n = \sqrt{\frac{x^\beta_n}2} \varphi_{n-1} + \sqrt{\frac{x^\beta_{n+1}}2} \varphi_{n+1}\; 
\label{eq:pos-op-act2}
\end{equation}
which is a two-term recursion relation, familiar from the theory of orthogonal polynomials. It follows that
$d\omega_{\beta}(x) = d\langle \varphi_0\vert E^\beta_x\vert \varphi_0\rangle $, and the $\varphi_n$ may be realized as the polynomials obtained
by orthonormalizing the sequence of monomials $1, x, x^2 , x^3 , \ldots\; , $ with respect to this measure
(using a Gram-Schmidt procedure). Let us use the notation $p^\beta_n (x)$ to write the vectors $\varphi_n$, when they are
so realized, as orthogonal polynomials in $L^2 (\mathbb R , d\omega_{\beta}(x))$. Then, for any $dw_\beta$-measurable set
$\Delta \subset \mathbb R$,
\begin{equation}
  \langle \varphi_k  ,E(\Delta) \varphi_n\rangle = \int_{\Delta}  p^\beta_k (x)p^\beta_n (x)\;d\omega_{\beta}(x) ,
\label{eq:poly-basis}
\end{equation}
and
\begin{equation}
    \langle \varphi_k , \varphi_n\rangle = \int_{\mathbb R}  p^\beta_k (x)p^\beta_n (x)\; d\omega_{\beta}(x)= \delta_{k,n}\; .
\label{eq:poly-orthog}
\end{equation}
We close this subsection by illustrating this formalism for the case of the sequence $x_n^\beta :=n+\beta$ with $\beta \geq 0$. 
\begin{proposition}
A polynomials realization of the basis $\{\varphi_n\}_{n=0}^\infty$ is given by the associated Hermite polynomials
\begin{equation}
\varphi_n^{\beta}(x)=\frac{2^{-n/2}}{\sqrt{(\beta +1)_n}}H_n(x,\beta ) \label{H_m_beta and p_m_beta}
\end{equation}
whose orthogonality measure is 
\begin{equation}
d\omega_{\beta}(x)=\left(\sqrt{\pi}\Gamma(\beta +1)\right)^{-1}|D_{-\beta}(ix\sqrt{2})|^{-2}dx
\end{equation} where $D_{a}(\cdot)$ is the parabolic cylinder function.
\end{proposition}
\begin{proof}
This can be proved by comparing $(\ref{eq:pos-op-act2})$, where the sequence $x_n^\beta$ is chosen to be $n+\beta$, with the three-terms recurrence relation in (\cite{AW}, p.16): 
\begin{eqnarray}
H_{n+1}(x,\beta)=2xH_n(x,\beta)-2(n+\beta)H_{n-1}(x,\beta), \qquad H_{-1}(x,\beta )=0,\ H_0(x,\beta)=1 \label{Recu_H_m_b}
\end{eqnarray}
where $H_n(x,\beta )$ are the associated Hermite polynomials.
\end{proof}
\begin{remark}
These polynomials were introduced and studied by Askey and Wimp in \cite{AW}. Their explicit orthogonality relation reads
\begin{equation}
\int_{\mathbb{R}}\frac{H_n(x,\beta)H_k(x,\beta)}{|D_{-\beta}(ix\sqrt{2})|^2}dx=2^n\sqrt{\pi}\Gamma(n+\beta +1)\delta_{n,k},\label{Orthog_H_m_beta}
\end{equation}
where $D_{a}$ is a parabolic cylinder function (\cite{Wunsche}, p199):
\begin{equation}
D_{\beta}(z)=\exp\left(-\frac{z^2}{4}\right)\frac{2^{\beta /2}\sqrt{\pi}}{\Gamma(-\beta)}\sum\limits_{k=0}^{+\infty}\frac{(-1)^k\Gamma(k-\beta)}{k!\Gamma\left(\frac{k-\beta+1}{2}\right)}\left(\frac{z}{\sqrt{2}}\right)^k.
\end{equation}
\end{remark}
\section{Generalized $\mu_{\beta}$-NLCS and their coherent states transforms}
In this section, we first review the construction of a general class of $2D$ orthogonal polynomials from \cite{MZ}. Then we introduce a generalized version of the NLCS $(\ref{NLCS_x_m,beta})$ with the use of these polynomials. We also discuss the associated coherent states transform and some related Hilbert spaces.
\subsection{A class of 2D orthogonal polynomials}
Let $d\mu_\beta $ be the measure given in \eqref{measure} and $\mu_\beta$ its moments with the normalization $\mu _{0}=1$. For each $\beta \geq 0$, let $\phi _{n}(r;\beta ),\;n=0,1,2,...$, be a
family of real polynomials, orthogonal with respect to the measure $d\mu
_\beta (r)$, that is 
\begin{equation}
\int_{0}^{L}\phi _{n}\left( r;\beta \right) \phi _{k}\left( r;\beta
\right) d\mu_\beta \left( r\right) =\zeta _{n}\left( \beta \right) \delta
_{k,n}\label{orth-rel of phi_n}  
\end{equation}%
where $\zeta _{n}(\beta )$ is a positive sequence. The polynomial $\phi _{n}(r;\beta )$ may also be written as 
\begin{equation}
\phi _{n}(r;\beta )=\sum_{j=0}^{n}c_{j}\left( n;\beta \right) r^{n-j},\quad \phi_0=1, \label{phi_n}
\end{equation}%
where the $c_{j}(n;\beta )$ are real coefficients.  Using the above real polynomials, the authors \cite{MZ} have constructed an orthogonal
family of polynomials, $P_{n,m}(z,\overline{z};\beta ),\; n,m=0,1,2,\ldots
$, in the variables $z,\overline{z}\in \mathbb{C}$, by 
\begin{eqnarray}
P_{n,m}(z,\overline{z};\beta )&=&z^{n-m}\phi _{m}(z\overline{z};n-m+\beta
),\quad n\geq m  \label{P_m,n}\\
&=&P_{m,n}(\overline{z},z;\beta ),\quad m\geq n. \label{P_n,m}
\end{eqnarray}
Note that $(\ref{P_m,n})$-$(\ref{P_n,m})$ can together be put into a single expression as
\begin{equation}
P_{n,m}(z,\overline{z};\beta )=z^{n-(n \wedge m)}\overline{z}^{m-(n \wedge m)}\phi _{n \wedge m}(z\overline{z};|n-m|+\beta
), \quad n,m\in \mathbb{Z}_+,\label{P_mn for all m,n}
\end{equation}
where $n\wedge m$ denotes the smaller of $n$ and $m$. Mourad and Zhang proved that these polynomials form
an orthogonal family, in the sense that 
\begin{equation}
\frac{1}{2\pi}\int_{0}^{L^{2}}\int_{0}^{2\pi }P_{n,m}(z,\overline{z};\beta )%
\overline{P_{k,s}(z,\overline{z};\beta )}d\theta d\mu_\beta ( r^{2}
) =\zeta _{n\wedge m}\left( \left\vert n-m\right\vert +\beta \right)\delta _{n,k}
\delta _{m,s}  \label{ortho_P_n,m}
\end{equation}%
where $\zeta _{n}(\beta) $ is as in $\left( \ref{orth-rel of phi_n}\right) $. We denote the inner product defined on the Hilbert space $L^{2,\beta}(\mathcal{D}_L):= L^2 (\mathcal{D}_L,  (2\pi)^{-1}d\theta d\mu_\beta(r^2))$ by 
\begin{equation}
\langle f,g\rangle_{L^{2,\beta}(\mathcal{D}_L)}=\frac{1}{2\pi}\int_{\mathcal{D}_L}f(z,\bar{z})\overline{g(z,\bar{z})}d\theta d\mu_\beta (z\bar{z}),\quad z=re^{i\theta}.
\end{equation}

We note that the biorthogonal polynomials $P_{n,m}(z,\overline{z};\beta ), \; n,m = 0,1,2, ... $, defined in (\ref{P_m,n}), span the Hilbert space
 $L^{2,\beta}(\mathcal{D}_L)$. Its holomorphic and antiholomorphic subspaces, $\mathcal{H}^\beta_{\text{hol}}(\mathcal{D}_L)$ and $\mathcal{H}^\beta_{\text{a-hol}}(\mathcal{D}_L)$, are spanned by the monomials $P_{n,0}(z,\overline{z},\beta)$, and $P_{0,n}(z,\overline{z},\beta), \;\; n=0,1,2, ... $, respectively. The above two monomial bases are in fact determinative of the entire set of polynomials
$P_{n,m}(z,\overline{z},\beta)$ since they determine the measure $d\mu_\beta$, through the moment relation
\begin{eqnarray}
\int_0^{L^2} \vert P_{n,0}(z,\overline{z},\beta)\vert^2\; d\mu_\beta (r^2) = \int_0^{L^2} r^{2n}\; d\mu_\beta(r^2) = \mu_{n+ \beta}\; ,
\end{eqnarray}
and hence of the entire Hilbert space $L^{2,\beta}(\mathcal{D}_L)$. 

\subsection{Generalized $\mu_{\beta}$-NLCS}
Keeping the same vectors basis $\varphi_n$ in the Hilbert space $\mathcal{H}$ as in the previous sections we can use polynomials $(\ref{P_m,n})$ to introduce a generalized version of the NLCS $(\ref{NLCS_x_m,beta})$ as follows. From now on, we will use the notation
\begin{equation}
P_{n,m}^{\beta}(z,\overline{z} ):=(\zeta _{0}(m+\beta ))^{-1/2}P_{n,m}(z,\overline{z}%
;\beta ),\;\; n,m\geq 0\label{P_m,n^beta}
\end{equation}
and we define the sequence $(x_{n,m}^{\beta})_{n\geq 0}$ by
\begin{equation}
x_{n,m}^{\beta}:=\frac{\zeta _{n\wedge m}(|n-m|+\beta )}{\zeta _{(n-1)\wedge m}(|n-m-1|+\beta )},
\end{equation}
with the generalized factorial
\begin{equation}
x_{n,m}^{\beta }!:=x_{n,m}^{\beta }x_{n-1,m}^{\beta }\ldots x_{1,m}^{\beta }=\frac{\zeta _{n\wedge m}(|n-m|+\beta )}{\zeta _{0}(m+\beta )},\text{ \ \ }x_{0,m}^{\beta }!\equiv 1\text{.}\quad m=1,2,3,...,\label{Gen_fact}
\end{equation}
and $\zeta_n(\alpha)$ is the coefficient in $(\ref{orth-rel of phi_n})$. By \eqref{P_m,n^beta} and \eqref{Gen_fact} the orthogonality relation $(\ref{ortho_P_n,m})$ takes the form
\begin{equation}
\frac{1}{2\pi}\int_{0}^{L^{2}}\int_{0}^{2\pi }P_{n,m}^\beta(z,\overline{z} )%
\overline{P_{k,s}^\beta(z,\overline{z})}d\theta d\mu_\beta ( r^{2}
) =x_{n,m}^{\beta }!\; \delta _{n,k}\delta _{m,s} . \label{ortho_P_n,m,beta}
\bigskip
\end{equation}
\begin{definition}
For fixed parameters $\beta \geq 0$ and $m\in\mathbb{Z}_+$, we define a set of generalized nonlinear
coherent states ($\mu_\beta$-GNLCS) as
\begin{equation}
\vartheta_{z,m,\beta}\equiv |z;\beta ,m\rangle :=\left( \mathcal{N}_{\beta,m}(z\overline{z})\right) ^{-\frac{1}{2}%
}\sum_{n=0}^{\infty }\frac{\overline{P_{n,m}^\beta(z,\overline{z})}}{\sqrt{%
 \;x_{n,m}^{\beta}!}}\varphi_n \label{GNLCS1}
\end{equation}
where $\mathcal{N}_{\beta,m}(z\overline{z})$ is a normalizing factor.
\end{definition}
We now give a condition on $z$ that ensures the finitness of $\mathcal{N}_{\beta,m}(z\overline{z})$. See Appendix A for the proof.
\begin{proposition}
The complex numbers $z$ for which the $\mu_\beta$-GNLCS $(\ref{GNLCS1})$ are defined belong to the disk $\mathcal{D}_{\beta ,m}=\{\xi\in\mathbb{C},\,\, |\xi|<R_{\beta ,m}\} $ with $R_{\beta ,m}:=\min\limits_{0\leqslant i,j \leqslant m} R_{\beta,m,i,j}$ where 
\begin{equation}
(R_{\beta, m,i,j})^2=\lim\limits_{n\rightarrow +\infty}\left|\frac{c_i(m;n-1+\beta)c_j(m;n-1+\beta) \zeta _{m}( n +\beta )}{c_i(m;n+\beta)c_j(m;n+\beta)\zeta _{m}( n -1+\beta )}\right| .\label{R_beta,n,i,j}
\end{equation}
In particular, for $m=0$, we have $\mathcal{D}_{\beta ,0}=\mathcal{D}_{\beta}$ is just the domain of convergence of $(\ref{Norm_x_m,beta})$.
\end{proposition}
For $m=0$, the polynomials $(\ref{P_m,n^beta})$ reduce to 
$P_{n,0}^{\beta}(z,\overline{z})=z^{n}/\sqrt{\mu_{\beta}},$
and $x_{n,0}^\beta !=\zeta_0(n+\beta)/\zeta_0(\beta)=\mu_{n+\beta}/\mu_{\beta}= x_n^\beta !\quad n=0,1,2,...$ . Consequently, the $\mu_\beta$-GNLCS $(\ref{GNLCS1})$ reduce to the $\mu_\beta$-NLCS in $(\ref{NLCS_x_m,beta})$.
\begin{proposition}
Assuming that $\mathcal{D}_L \subseteq \mathcal{D}_{\beta ,m} $, then the $\mu_\beta$-GNLCS in $(\ref{GNLCS1})$ satisfy the resolution of the identity operator of $\mathcal{H}$ as
\begin{equation}
\int_{\mathcal{D}_L}|\vartheta_{z,m,\beta}\rangle \langle \vartheta_{z,m,\beta} | d\eta_{\beta ,m} (z,\bar{z})=\mathbf{1}_{
\mathcal{H}}.\label{RI-Gener}
\end{equation}
where $d\eta_{\beta ,m} (z,\bar{z})=(2\pi)^{-1}d\theta \mathcal{N}_{\beta ,m}(z\overline{z} )d\mu_\beta (z\bar{z} )$.
\end{proposition}
\begin{proof}
Let us assume that the measure takes the form $d\eta_{\beta,m} (z,%
\bar{z})=\mathcal{N}_{\beta ,m}(z\bar{z})h(z\bar{z})d\nu(z),$
where $h$ is an auxiliary density function to be determined and $d\nu(z)$ is a Lebesgue measure on $\mathbb{C}$. In terms of
polar coordinates $z=\rho e^{i\theta} ,\ \rho>0$ and $\theta\in [0,2\pi)$, the measure can be rewritten as $d\eta_{\beta ,m} (z,%
\bar{z})=\pi^{-1} \mathcal{N}_{\beta ,m}(\rho^2 )h(\rho^2)\rho d\rho d\theta.$ Using the expression $(\ref{NLCS_x_m,beta})$ of coherent states, the operator
\begin{equation}
\mathcal{O}_{\beta ,m}=\int_{\mathcal{D}_L}\left| \vartheta_{z,m,\beta}\right\rangle \left\langle \vartheta_{z,m,\beta}\right| d\eta_{\beta ,m} (z,\bar{z})
\end{equation}
reads,
\begin{equation}
\mathcal{O}_{\beta,m}=\sum\limits_{s,n=0}^{+\infty} \frac{(\pi\zeta _{0}(m+\beta ))^{-1}}{\sqrt{x_{s,m}^\beta !}\sqrt{x_{n,m}^\beta !}}\left(  \int_0^{L^2}\int_0^{2\pi}P_{s,m}(z,\overline{z},\beta)\overline{P_{n,m}(z,\overline{z},\beta)} d\theta h(\rho^2)\rho d\rho \right)  \vert \varphi_n\rangle \langle \varphi_s\vert \label{O_beta,n}
\end{equation}
The orthogonality relation $(\ref{ortho_P_n,m})$  suggests us to choose the function $h(r)$ such that $h(r)=d\mu_{\beta}(r)/dr$ i.e., the density of $d\mu_{\beta}(r)$ with respect to $dr$. Since
\begin{equation}
(2\pi\zeta _{0}(m+\beta ))^{-1}\int_{0}^{L^{2}}\int_{0}^{2\pi }P_{s,m}(z,\overline{z};\beta )%
\overline{P_{n,m}(z,\overline{z};\beta )}d\theta d\mu_\beta ( \rho^{2}
) =x_{n,m}^\beta !\;
\delta _{s,n} 
\end{equation}
and $\{\varphi_n\}$ is an orthonormal basis of $\mathcal{H}$, equation $(\ref{O_beta,n})$ reduces to $\mathcal{O}_{\beta,m}=\sum_{n=0}^{+\infty}\left|\varphi_n\right\rangle\left\langle \varphi_n\right|=\textbf{1}_{\mathcal{H}}$.
\end{proof}
\subsection{The coherent states transform and its range}
The resolution of identity $(\ref{RI-Gener})$ lead us to define the associated coherent states transform, which maps the Hilbert space $\mathcal{H}$ onto a subspace of the Hilbert space $L^{2,\beta}(\mathcal{D}_L)$, as follows.

\begin{proposition}
The generalized Bargmann transform associated with the coherent states $\vartheta_{z,m,\beta}$ is the isometric map $\mathcal{B}_{m,\beta}:\mathcal{H}\longrightarrow L^{2,\beta}(\mathcal{D}_L)$,
defined by
\begin{equation}
 \mathcal{B}_{m,\beta}\left[ \phi\right] (z):=\left( \mathcal{N}_{m,\beta}(z\overline{z} )\right) ^{\frac{1}{2}%
}\langle \phi ,\vartheta_{z,m,\beta}\rangle_{\mathcal{H}} .\label{W_beta,n}
\end{equation} 
Thus, for $\phi,\ \psi \in \mathcal{H}$, we have 
\begin{equation}
\langle \phi ,\psi\rangle_{\mathcal{H}}=\langle \mathcal{B}_{m,\beta}\left[ \phi\right] , \mathcal{B}_{m,\beta}\left[ \psi\right]\rangle_{L^{2,\beta}(\mathcal{D}_L)}.
\end{equation}
In particular,
\begin{equation}
\mathcal{B}_{m,\beta}
\left[\varphi_n\right] (\overline{z})=\frac{P_{n,m}^\beta(z,\overline{z})}{\sqrt{x_{n,m}^{\beta}!}},\quad n=0,1,2,\cdots \, .
\end{equation}
\end{proposition}
Note that the range of $\mathcal{B}_{m,\beta}$ is the span of polynomials $\left(P_{n,m}^\beta(z,\overline{z})/\sqrt{x_{n,m}^{\beta}!}\right)_{n\geq 0}$ which we denote by $\mathcal{A}^{2}_{\beta ,m}(\mathcal{D}_L)$. This subspace of $L^{2,\beta}(\mathcal{D}_L)$ can be compared with the eigenspace of a specific operator that we construct in the following way. We start by denoting 
\begin{equation}
\widetilde{P}_{n,m}^{\beta }(z,\overline{z}):=\frac{P_{n,m}^{\beta}(z,\overline{z} )}{\sqrt{ \;x_{n,m}^{\beta}!}}\text{,}
 \label{P_n,m,b} 
\end{equation}
which form an orthonormal basis in the full Hilbert space $L^{2,\beta}(\mathcal{D}_L)$. Next, we define two pairs of operators, $A_i^\beta, \; A_i^{\beta\dag}\; \; i= 1,2$, and some related operators by using the orthonormalized polynomials $(\ref{P_n,m,b})$ as follows.  
\begin{eqnarray}
  A^\beta_1 \widetilde{P}_{n,m}^{\beta } = \sqrt{x^\beta_{n,m}}\widetilde{P}_{n-1,m}^{\beta }, \;\; n = 1,2,3, \ldots , & \qquad & A^\beta_1 \widetilde{P}_{0,m}^{\beta } = 0, \nonumber\\
  A^\beta_2 \widetilde{P}_{n,m}^{\beta } = \sqrt{x^\beta_{m,n}}\widetilde{P}_{n,m-1}^{\beta }, \;\; m = 1,2,3, \ldots , & \qquad & A^\beta_2 \widetilde{P}_{n,0}^{\beta } = 0,
\label{global-cr-ann1}
\end{eqnarray}
and their adjoints
\begin{equation}
  A^{\beta\dag}_1 \widetilde{P}_{n,m}^{\beta } = \sqrt{x^\beta_{n + 1,m}}\widetilde{P}_{n+1,m}^{\beta }, \; n=0,1,2,... , \quad
  A^{\beta\dag}_2 \widetilde{P}_{n,m}^{\beta } = \sqrt{x^\beta_{m + 1,n}}\widetilde{P}_{n,m+1}^{\beta }, \;\; m = 0,1,2, ... .
\label{global-cr-ann2}
\end{equation} 
The two operators $A^\beta_1, A^{\beta\dag}_1$ commute with $A^\beta_2, A^{\beta\dag}_2$ so that, together with the identity operator $I$ on $L^{2,\beta}(\mathcal{D}_L)$. Recalling that $[ A_1^{\beta\dag},  A_2^{\beta\dag}] = 0$, the basis vectors $\widetilde{P}_{n,m}^{\beta }$ themselves may be written in terms of the operators
$A_i^{\beta\dag}, \; i = 1,2$, as
\begin{equation}
 \widetilde{P}_{n,m}^{\beta } = \frac {( A_1^{\beta\dag})^n\; (A_2^{\beta\dag})^m}{\sqrt{x^\beta_{n,m}!\;x^\beta_{m,0}!}} \widetilde{P}_{0,0}^{\beta } ,
\label{oprep-polyn}
\end{equation}
 where the initial vector $\widetilde{P}_{0,0}^{\beta }$ is just the constant function
$$
\widetilde{P}_{0,0}^{\beta }(z , \overline{z}) = 
\frac 1{\sqrt{\zeta_0(\beta)}} , \qquad  \;z \in \mathbb{C}. $$
\begin{definition}
Two operators are defined as
\begin{equation}
  L_i^{\beta} := \frac 12 (A_i^{\beta\dag}A_i^{\beta} + A_i^{\beta}A_i^{\beta\dag}), \qquad i = 1,2 .
\label{landau-ops}
\end{equation}
\end{definition}
By $(\ref{global-cr-ann1})-(\ref{global-cr-ann2})$ we easily see that their action on the basis vectors $P^\beta_{m,n}$ is given by
\begin{equation}
  L_1^\beta \widetilde{P}_{n,m}^{\beta } = \lambda_{m,n}^\beta\widetilde{P}_{n,m}^{\beta }, \qquad  L_2^\beta \widetilde{P}_{n,m}^{\beta } = \lambda_{n,m}^\beta\widetilde{P}_{n,m}^{\beta },
\label{landop2}
\end{equation}
where
\begin{equation}
\lambda_{m,n}^\beta=\frac 12 (x_{m+1,n}^\beta + x_{m,n}^\beta ).
\end{equation}
Clearly, $[L_1^\beta , L_2^\beta] = 0$, and the spectrum of each operator is infinitely degenerate.\\
\\
These operators are useful since, for example, the range of the coherent state transform $\mathcal{B}_{m,\beta}$ could be characterized as the eigenspace of $L_1^\beta$ corresponding to the eigenvalue $\lambda_{m,n}^{\beta}$. This statement will be illustrated below with a particular choice of the measure $d\mu_\beta(r)$.
\section{The example of the measure $r^{\beta}e^{-r}dr$}
We illustrate the above construction with the example of the measure $d\mu_\beta (r):=r^\beta e^{-r}dr$, which leads to a new generalization of the true-polyanalytic Bargmann transform. For this, we make use of the vectors basis $\varphi_n^\beta $ in the Hilbert space $L^2(\mathbb{R},\; d\omega_{\beta}(x))$ as defined by \eqref{H_m_beta and p_m_beta}.
\subsection{Case $m\neq 0$} According to \cite{MZ} the measure $r^{\beta}e^{-r}dr$ gives rise to the polynomials $\phi _{n}(r;\beta )=(-1)^n L_{n}^{\left( \beta \right) }\left( r\right) ,$ where $L_{n}^{\left( \beta \right) }$ is the Laguerre polynomial (\cite{askey}, p.242). It has been shown \cite{MZ} that the resulting complex polynomials from $(\ref{P_m,n})$ by the above procedure are
\begin{equation}
H_{n,m}^{(\beta )}(z,\overline{z})=\frac{1}{m!}\sum_{k=0}^{m}\left( 
\begin{array}{c}
m \\ 
k
\end{array}
\right) \frac{\left( \beta +1\right) _{n}}{\left( \beta +1\right) _{n-k}}%
\left( -1\right) ^{k}z^{n-k}\overline{z}^{m-k}, \label{GHPoly}
\end{equation}
for $n\geq m$. In other words, 
\begin{equation}
H_{n,m}^{(\beta )}(z,\overline{z})=(-1)^m z^{n-m}L_{m}^{\left(
\beta +n-m\right) }\left( z\overline{z}\right) ,\quad n\geq m.  
\end{equation}
When $n<m$ they are defined by $H_{n,m}^{(\beta )}(z,\overline{z}
)=H_{m,n}^{(\beta )}(\overline{z},z)$. In \cite{MZ} these polynomials were denoted by $Z_{m,n}^{(\beta)}(z,\overline{z})$. These polynomials can be rewritten for all $n,m\in \mathbb{Z}_+$ as
\begin{equation}
H_{n,m}^{(\beta )}(z,\overline{z})=(-1)^{n \wedge m}|z|^{|n-m|}e^{i(n-m)\arg z}L_{n \wedge m}^{(|n-m|+\beta)}(z\overline{z}).\label{GCHP}
\end{equation}
 The function $\zeta _{n}\left( \beta \right)$ in $\left( \ref{orth-rel of phi_n}\right) $ together with coefficients $c_{j}\left( n;\beta \right) $ in $\left( \ref{phi_n}\right) $ were given by 
\begin{equation}
\zeta _{n}\left( \beta\right) =\frac{\Gamma \left( \beta +n+1\right) }{n!}%
,\text{ \ \ }c_{j}\left( n;\beta\right) =\frac{\left( -1\right)
^{n-j}\left( \beta +1\right) _{n}}{\left( n-j\right) !j!\left( \beta
+1\right) _{n-j}}\text{.} 
\end{equation}
The orthogonality relation satisfied by these polynomials is
\begin{eqnarray}
\int_{\mathbb{C}}H_{j,n}^{\left( \beta \right) }\left( z,\overline{z}\right)\overline{H_{l,m}^{\left( \beta \right) }\left( z,\overline{z}\right)}(z\bar{z})^{\beta}e^{-z\bar{z}} rdrd\theta = \pi\frac{ \Gamma \left( \beta +j\vee n+1\right)}{\left(j\wedge n\right) ! }\delta_{j,l}\delta_{n,m}.\label{OR for HP}
\end{eqnarray}
\begin{remark}
For $\beta =0$, the polynomials in $(\ref{GHPoly})$ reduces to the complex Hermite polynomials: 
\begin{equation}
H_{m,n}(z,\overline{z})=\left( m\wedge n\right) !H_{m,n}^{(0)}(z,\overline{z}%
)=\sum_{k=0}^{m\wedge n}\left( 
\begin{array}{c}
m \\ 
k%
\end{array}%
\right) \left( 
\begin{array}{c}
n \\ 
k%
\end{array}%
\right) \left( -1\right) ^{k}k!z^{m-k}\overline{z}^{n-k}  \label{2.20}
\end{equation}%
which were first introduced by Ito \cite{Ito}. They are orthogonal with respect to the weight $e^{-x^2-y^2}$ on $\mathbb R^2$. It must be noted that a more general class of polynomials, called  multidimensional Hermite polynomials.  Their weight function is $\exp(- \sum_{j,k} c_{j,k} x_j x_k)$ on $\mathbb R^n$, where
the quadratic form $\sum_{j,k} c_{j,k} x_j x_k$ is assumed to be positive definite.  Although the   more general
multidimensional Hermite polynomials were known long before the Ito $2D$ Hermite polynomials, the structure of the multidimensional Hermite polynomials is much more 
complicated and very little is still known about them. On the other hand the $2D$-Hermite polynomials are expressed in terms of  $z$ and $\bar z$ instead of $x$ and $y$ resulting in a major simplification of the results.
\end{remark}
In this case, the polynomials in $(\ref{GCHP})$ reads
\begin{equation}
P_{n,m}^{\beta}(z,\bar{z})=(\Gamma(\beta +m+1))^{-\frac{1}{2}}H_{n,m}^{(\beta )}(z,\bar{z}),
\end{equation}
the generalized factorial in $(\ref{Gen_fact})$ becomes
\begin{equation}
x_{n,m}^{\beta}!=\frac{ \Gamma \left( \beta +n\vee m+1\right)
}{\Gamma(\beta +m+1)\left( n\wedge m\right) !}
\end{equation}
and the definition of GNLCS takes the following form.
\begin{definition}
Let $\beta \geq 0$ and $%
m=0,1,2,...,\infty $. The GNLCS can be defined throughout the following superposition
\begin{equation}
\vartheta_{z,m,\beta} :=\left(\mathcal{N}_{\beta,m}(z\overline{z})\right) ^{-\frac{1}{2}%
}\sum_{n=0}^{\infty }\frac{\overline{H_{n,m}^{\left( \beta \right) }\left( z,%
\overline{z}\right) }}{\sqrt{\frac{ \Gamma \left( \beta +n\vee m+1\right)
}{\left( n\wedge m\right) !}}}\varphi_n^\beta \label{GNLCS}
\end{equation}
where $\mathcal{N}_{\beta,m}(z\overline{z})$ is a normalization factor. Here $\{\varphi_n^ \beta\} $ are given by \eqref{H_m_beta and p_m_beta}.
\end{definition}
Now, we give the overlap between two GNLCS. See Appendix B for the proof.

\begin{proposition}
The overlap is given by
\begin{eqnarray*}
\langle \vartheta_{z,m,\beta}|\vartheta_{w,m,\beta}\rangle &=&\left( \mathcal{N}_{\beta ,m}(z\overline{z})\right) ^{-\frac{1}{2}}\left(\mathcal{N}_{\beta ,m}(w\overline{w})\right) ^{-\frac{1}{2}}\left[ \sum\limits_{j=0}^{m-1 } \frac{j !(\overline{z}w)^{m-j}}{\Gamma \left( \beta + m+1\right) }
L_{j}^{\left( \beta +m-j\right) }\left( z\overline{z}\right) L_{j}^{\left( \beta +m-j\right) }\left( w\overline{w}\right) \right.\\
&&\left. +\frac{(\beta+1)_m}{m!\Gamma(\beta+1)}\sum\limits_{k=0}^{m}\sum\limits_{l=0}^{m}\frac{(-m)_k(-m)_{l}(z\overline{z})^k(w\overline{w})^{l}}{k!l!(\beta+1)_k(\beta+1)_{l}}{}_{2}F_{2}\left( 
\begin{array}{c}
1,m+\beta +1 \\ 
k+\beta +1,l+\beta +1
\end{array}%
\big|z\overline{w}\right)\right].
\end{eqnarray*}
in terms of the ${}_2F_2$ hypergeometric series.
\end{proposition}
\begin{proposition}
The GNLCS $(\ref{GNLCS})$ satisfy the following resolution of the identity
\begin{equation}
\int_{\mathbb{C}}\left| \vartheta_{z,m,\beta}
\rangle \langle \vartheta_{z,m,\beta}\right| d\eta_{\beta ,m}(z)=\textbf{1}_{\mathcal{H}},
\end{equation}
where
\begin{eqnarray*}
d\eta_{\beta ,m}(z)&=&\left(\Gamma(\beta +1)\right)^{-1}\left[ \sum\limits_{j=0}^{m-1} \frac{j !(z\bar{z})^{m-j}}{\Gamma \left( \beta + m+1\right) }
\left[ L_{j}^{\left( \beta +m-j\right) }\left( z\bar{z}\right) \right]^{2}\right.\\
  &+& \left.\frac{(\beta+1)_m}{m!\Gamma(\beta+1)}\sum\limits_{k,l=0}^{m}\frac{(-m)_k(-m)_{l}\; (z\bar{z})^{k+l}}{k!l!(\beta+1)_k(\beta+1)_{l}}{}_{2}F_{2}\left( 
\begin{array}{c}
1,m+\beta +1 \\ 
k+\beta +1,l+\beta +1
\end{array}
\big|z\bar{z}\right)\right](z\bar{z})^{\beta}e^{-z\bar{z}}d\nu(z),
\end{eqnarray*}
in terms of the ${}_2F_2-$series and the Lebesgue measure $d\nu $.
\end{proposition}
\begin{proof}
The resolution of the identity of the Hilbert space $\mathcal{H}$ follows from the orthogonality relation $(\ref{OR for HP})$ of $H_{n,m}^{\left( \beta \right) }\left( z,\overline{z}\right)$ by using similar arguments as in the proof of Proposition 3.2.2. The normalization is deduced from Proposition 4.1.1. by taking $z=w$ with the condition $\langle \vartheta_{z,m,\beta}|\vartheta_{z,m,\beta}\rangle =1$.
\end{proof}
Now, we can define the associated coherent states transform which maps the Hilbert space $L^2(\mathbb{R},\; d\omega_{\beta}(x))$ onto the Hilbert space $\mathcal{A}_{\beta ,m}^2(\mathbb{C})$ of complex-valued functions on $\mathbb{C}$, subspace of the larger Hilbert space $L^2(\mathbb{C},d\eta_{\beta ,m}(z))$ of square integrable functions with respect to the measure $d\eta_{\beta ,m}$ in \eqref{RI-Gener}.
\begin{theorem}
The GNLCS (\ref{GNLCS}) give rise to a generalized Bargmann transform through the isometric embedding $\mathcal{B}_{m,\beta}
:L^{2}(\mathbb{R},d\omega_{\beta}(x))\longrightarrow \mathcal{A}_{\beta ,m}^2(\mathbb{C})$ defined by
\begin{equation}
 \mathcal{B}_{m,\beta}[\varphi ](z)=\int_{\mathbb{R}}B_{\beta,m}(z,x)\varphi(x)d\omega_{\beta}(x),\label{BAR}
\end{equation}%
where
\begin{eqnarray}
B_{\beta,m}(z,x)=\sum\limits_{n=0}^{m-1}\frac{2^{n/2}}{\sqrt{(\beta +1)_n}}H_n(x,\beta)\left[\frac{(-1)^n z^{m-n}\sqrt{n!}}{\sqrt{\Gamma(\beta +m+1)}}L_n^{(m-n+\beta)}(z\bar{z})-\frac{(-1)^m \bar{z}^{n-m}\sqrt{m!}}{\sqrt{\Gamma(\beta +n+1)}}L_m^{(n-m+\beta)}(z\bar{z}) \right] \nonumber \\ +\frac{z^m}{\sqrt{\Gamma \left( \beta +1\right)m!}}\sum_{k=0}^{m}\frac{(-m)_k(\beta +1-k)_{k}}{k!\left(z\bar{z}\right)^{k}}F_{1:0;0;0}^{1:0;0;1}\left(
\begin{array}{c}
[1:1,2,1]:-;-;[\beta :1] \\ \left[\beta -k+1:1,2,2\right]:-;-;-
\end{array}
 \sqrt{2}x\overline{z}, -\overline{z}^2/2, -\overline{z}^2\right).\label{B_beta,m}
\end{eqnarray}
\end{theorem}
The special function $F_{1:0;0;0}^{1:0;0;1}$ in the right hand side of the last equation is the generalized Lauricella function, see Appendix C for its definition and details. The proof of Theorem 4.1.1. is given in Appendix D.\\
When $\beta =0$, the measure $d\omega_{0}(x)=\pi^{-\frac{1}{2}} e^{-x^2}dx$ and $\mathcal{A}_{0 ,m}^2(\mathbb{C})$ turns out to be the space of true-$m$-polyanalytic functions that is the orthogonal difference $\mathfrak{F}_{m}^2(\mathbb{C})\circleddash \mathfrak{F}_{m-1}^2(\mathbb{C})$ between two consecutive $m$-analytic spaces 
\begin{equation}
\mathfrak{F}_{m}^2(\mathbb{C})=\{g:\mathbb{C}\rightarrow \mathbb{C},\, \int_{\mathbb{C}} |g(z)|^2e^{-|z|^2}d\nu(z)<+\infty ,\, \overline{\partial}^m g(z)=0 \}
\end{equation}
The space $\mathcal{A}_{0 ,m}^2(\mathbb{C})$ can also be realized as eigenspace in $L^2\left(\mathbb{C}, e^{-|z|^2}d\nu(z)\right)$ of the operator $\widetilde{\Delta }_0$ given by \eqref{delta0} and corresponding to the eigenvalue $m$.
\begin{corollary}
For $\beta =0$, the integral transform \eqref{BAR} reduces to the true-polyanalytic Bargmann transform $\mathcal{B}_{m,0}
:L^{2}(\mathbb{R},\pi^{-\frac{1}{2}} e^{-x^2}dx)\longrightarrow \mathcal{A}_{0 ,m}^2(\mathbb{C})$ given by
\begin{equation}
 \mathcal{B}_{m,0}[\varphi ](z)=\pi^{-\frac{1}{2}}\int_{\mathbb{R}}B_{0,m}(z,x)\varphi(x)e^{-x^2}dx,\label{BAR0}
\end{equation}%
where
\begin{equation}
B_{0,m}(z,x)=(-1)^m(2^m m!)^{-\frac{1}{2}}e^{\sqrt{2}x\overline{z}-\frac{1}{2}\bar{z}^2} H_m\left(x-\frac{z+\bar{z}}{\sqrt{2}}\right).
\end{equation}
\end{corollary}
\begin{proof}
If we take $\beta =0$ in \eqref{BAR} we get 
\begin{eqnarray}
B_{0,m}(z,x)=\sum\limits_{n=0}^{m-1}\frac{2^{n/2}}{\sqrt{n!}}H_n(x)\left[\frac{(-1)^n z^{m-n}\sqrt{n!}}{\sqrt{m!}}L_n^{(m-n)}(z\bar{z})-\frac{(-1)^m \bar{z}^{n-m}\sqrt{m!}}{\sqrt{n!}}L_m^{(n-m)}(z\bar{z}) \right] \nonumber \\  +\frac{z^m}{\sqrt{m!}}\sum_{k=0}^{m}\frac{(-m)_k(1-k)_{k}}{k!\left(z\bar{z}\right)^{k}}F_{1:0;0}^{1:0;0}\left(
\begin{array}{c}
[1:1,2]:-;- \\ \left[1 -k:1,2\right]:-;-
\end{array}
 \sqrt{2}x\overline{z}, -\overline{z}^2/2\right).\label{B_0}
\end{eqnarray}
Making use of the identity (\cite{Sz}, p. 102):
\begin{equation}
L_m^{(-k)}(t)=(-t)^k\frac{(m-k)!}{m!}L_{m-k}^{(k)}(t),\quad 1\leq k\leq m, 
\end{equation}
for $k=n-m$ and $t=z\bar{z}$, we find that the first sum in \eqref{B_0} is zero.\\ 
Denoting the second sum in \eqref{B_0} by $S_m$, straightforward calculations give successively  
\begin{eqnarray*}
(-1)^m\frac{ 2^{\frac{1}{2}m}}{\sqrt{m!}} S_m &=&\frac{ 2^{\frac{1}{2}m}}{m!}(-z)^m\sum_{k=0}^{m}\frac{(-m)_k(1-k)_{k}}{k!\left(z\bar{z}\right)^{k}}F_{1:0;0}^{1:0;0}\left(
\begin{array}{c}
[1:1,2]:-;- \\ \left[1 -k:1,2\right]:-;-
\end{array}
 \sqrt{2}x\overline{z}, -\overline{z}^2/2\right)\\ &=& \frac{ 2^{\frac{1}{2}m}}{m!}(-z)^m\sum_{k=0}^{m}\frac{(-m)_k(1-k)_{k}}{k!\left(z\bar{z}\right)^{k}}\sum_{i,j=0}^\infty \frac{(1)_{i+2j}}{(1-k)_{i+2j}}\frac{(\sqrt{2}x\overline{z})^i}{i!}\frac{(-\overline{z}^2/2)^j}{j!}\\
 &=& \frac{ 2^{\frac{1}{2}m}}{m!}(-z)^m\sum_{k=0}^{m}\frac{(-m)_k}{k!\left(z\bar{z}\right)^{k}}\sum_{i,j=0}^\infty (i+2j+1-k)_{k}\frac{(\sqrt{2}x\overline{z})^i}{i!}\frac{(-\overline{z}^2/2)^j}{j!}\\
 &=& \left(\frac{\sqrt{2}}{\bar{z}}\right)^{m}\sum_{i,j=0}^\infty \frac{(\sqrt{2}x\overline{z})^i}{i!}\frac{(-\overline{z}^2/2)^j}{j!}\sum_{k=0}^{m}\frac{(-m)_k}{k!m!}(i+2j+1-k)_{k}(-z\bar{z})^{m-k}\\
 &=&  \left(\frac{\sqrt{2}}{\bar{z}}\right)^{m}\sum_{i,j=0}^\infty \frac{(\sqrt{2}x\overline{z})^i}{i!}\frac{(-\overline{z}^2/2)^j}{j!}L_{m}^{(i+2j-m)}(z\bar{z}).
\end{eqnarray*}
Next, we multiply $(-1)^m\frac{ 2^{\frac{1}{2}m}}{\sqrt{m!}}S_m$ by $t^m$ and we sum over $m$, as
\begin{equation}
\sum_{m=0}^\infty (-1)^m\frac{ 2^{\frac{1}{2}m}}{\sqrt{m!}}S_m \,t^m =\sum_{i,j=0}^\infty \frac{(\sqrt{2}x\overline{z})^i}{i!}\frac{(-\overline{z}^2/2)^j}{j!}\sum_{m=0}^\infty (\sqrt{2} t/\bar{z})^m L_{m}^{(i+2j-m)}(z\bar{z}).
\end{equation}
We now make use of the generating function of Laguerre polynomials \cite{Carlitz1960}
\begin{equation}
\sum_{m=0}^\infty L_{m}^{(\alpha-m)}(x)s^m=(1+s)^{\alpha }e^{-xs},\qquad |s|<1,
\end{equation}
for $\alpha =i+2j$, $s=\sqrt{2} t/\bar{z}$ and $x=z\bar{z}$ we obtain 
\begin{eqnarray}
\sum_{m=0}^\infty (-1)^m\frac{ 2^{\frac{1}{2}m}}{\sqrt{m!}}S_m \,t^m &=& \sum_{i,j=0}^\infty  e^{-\sqrt{2}zt}\frac{(\sqrt{2}x\overline{z}+2xt)^i}{i!}\frac{\left(-(\bar{z}+\sqrt{2} t)^2/2\right)^j}{j!}\\
&=& \sum_{m=0}^\infty \frac{1}{m!}e^{\sqrt{2}x\overline{z}-\frac{\bar{z}^2}{2}}H_m\left(x-\frac{z+\bar{z}}{\sqrt{2}}\right) \,t^m.
\end{eqnarray}
Then, by comparing coefficients of identical power of $t$, we obtain 
\begin{equation}
S_m =\frac{1}{m!}e^{\sqrt{2}x\overline{z}-\frac{\bar{z}^2}{2}} H_m\left(x-\frac{z+\bar{z}}{\sqrt{2}}\right).
\end{equation}
Finally, 
\begin{equation}
B_{0,m}(z,x)=(-1)^m 2^{-\frac{m}{2}}\sqrt{m!} S_m=(-1)^m (2^m m!)^{-\frac{1}{2}} e^{\sqrt{2}x\overline{z}-\frac{\bar{z}^2}{2}}H_m\left(x-\frac{z+\bar{z}}{\sqrt{2}}\right),
\end{equation}
which is the kernel function of the true-$m$-polyanalytic Bargmann transform.
\end{proof}
\subsection{The analytic case $m=0$} In this section, we particularize the above construction for the measure $d\mu_\beta(r):=r^{\beta}e^{-r}dr$. 
In this case the moments $\mu_{n+\beta}=\Gamma(n+\beta +1)$, the sequence $x_n^{\beta}=n+\beta$, $\lim\limits_{n\rightarrow \infty}x_n^\beta =+\infty$ and the resulting $\mu_\beta -$NLCS defined by $(\ref{NLCS_x_m,beta})$ take the form 
\begin{equation}
\vartheta_{z,\beta} =\left(\mathcal{N}_{\beta}(z\overline{z})\Gamma(\beta+1)\right)^{-\frac{1}{2}}\sum_{n=0}^{\infty }\ \frac{\overline{z}^{n}}{\sqrt{(\beta +1)_n}}
\varphi_n^\beta ,\label{CSs}
\end{equation}
where the normalizing factor is given by
\begin{equation}
\mathcal{N}_{\beta}(z\overline{z})=\frac{e^{z\overline{z}}}{\Gamma(\beta+1)} {}_1F_1(\beta,\beta +1; -z\overline{z})\label{normalization_n=0}
\end{equation} 
in terms of the confluent hypergeometric ${}_1F_1-$sum. \\
\\
We now give a measure with respect to which the NLCS $(\ref{CSs})$ ensure the
resolution of the identity of $\mathcal{H}$ (see Appendix E for the proof).

\begin{proposition}
Let $\beta \geq 0$. Then, the NLCS $(\ref{CSs})$ satisfy the following
resolution of the identity
\begin{equation}
\int_{\mathbb{C}}\left| \vartheta_{z,\beta} 
\rangle \langle \vartheta_{z,\beta}\right| d\eta_{\beta}(z)=\textbf{1}_{\mathcal{H}},\label{reso}
\end{equation}
where
\begin{equation}
d\eta_{\beta}(z)=\left(\Gamma(\beta+1)\right)^{-1} {}_1F_1(\beta,\beta +1; -z\bar{z})(z\bar{z})^{\beta} d\nu(z).\bigskip\label{measure-RI}
\end{equation}
where $d\nu(z)$ is a Lebesgue measure on $\mathbb{C}$.
\end{proposition}
\begin{remark}
The above coherent states in $(\ref{CSs})$ are known in the literature as the Mittag-Leffler coherent states \cite{SPS}. Indeed, the measure $(\ref{measure-RI})$ also reads $d\eta(re^{i\theta},\beta)=\pi^{-1}E_{1,\beta +1}(r^2)r^{2\beta +1}e^{-r^2}drd\theta$, where 
\begin{eqnarray*}
E_{\alpha,\gamma}(x)= \sum_{n=0}^{\infty }\frac{x^{n}}{\Gamma(\alpha n+\gamma )}\quad \alpha, \gamma >0.\label{ML}
\end{eqnarray*}
For $\gamma =1$, these last expression is the well known Mittag-Leffler function \cite{ML}. Its generalization $(\ref{ML})$ was introduced later \cite{Wiman}.
\end{remark}
Once we have obtained a closed form for the NLCS (\ref{CSs}), we can define the associated coherent states transform. The latter one should map the Hilbert space $L^2(\mathbb{R},\; d\omega_{\beta}(x))$ onto the Hilbert space $\mathcal{A}_{\beta}^2(\mathbb{C})$ of complex-valued analytic functions on $\mathbb{C}$, subspace of the larger Hilbert space $L^2(\mathbb{C},d\eta_\beta(z))$ of square integrable functions with respect to the measure $d\eta_\beta$. From the expression $(\ref{CSs})$ and the general theory of coherent states \cite[p.75]{Gazeau2009} the overlapping function $\langle \vartheta_{z,\beta} |\vartheta_{w,\beta}\rangle$ between two coherent states provide us with the reproducing kernel of the space $\mathcal{A}_{\beta}^2(\mathbb{C})$ via the relation  
  \begin{eqnarray}
K_\beta(z,w)&=&\left(\mathcal{N}_{\beta}(z\overline{z})\right)^{\frac{1}{2}}\left(\mathcal{N}_{\beta}(w\overline{w})\right)^{\frac{1}{2}}\langle \vartheta_{z,\beta} |\vartheta_{w,\beta}\rangle\\
&=& \frac{e^{z\overline{w}}}{\Gamma(\beta+1)} {}_1F_1(\beta,\beta +1; -z\overline{w}),\quad z,w\in\mathbb{C}.\label{RK_n=0}
\end{eqnarray} 
The following result makes this statement more precise. 
\begin{theorem}
The NLCS (\ref{CSs}) give rise to a generalized Bargmann transform through the unitary embedding $\mathcal{B}_{\beta}
:L^{2}(\mathbb{R},d\omega_{\beta}(x))\rightarrow \mathcal{A}_{\beta}^2(\mathbb{C})$ defined by
\begin{equation}
 \mathcal{B}_{\beta}[\varphi ](z)=\int_{\mathbb{R}}B_{\beta}(z,x)\varphi(x)d\omega_{\beta}(x),\label{BAR}
\end{equation}%
where 
\begin{eqnarray}
B_{\beta}(z,x)=F_{1:0;0;0}^{1:0;0;1}\left(
\begin{array}{c}
[1:1,2,1]:-;-;[\beta :1] \\ \left[\beta+1:1,2,2\right]:-;-;-
\end{array}
 \sqrt{2}x\overline{z}, -\overline{z}^2/2, -\overline{z}^2\right).\label{Lauricella}
\end{eqnarray}
\textit{In particular, when $\beta =0$}
\begin{equation}
B_0(z,x)=e^{-\frac{1}{2}\bar{z}^2+\sqrt{2}x\bar{z}},\quad z\in\mathbb{C},\ x\in\mathbb{R},\label{B_00}
\end{equation}
and $\mathcal{B}_0$ is the well known classical Bargmann transform acting on function in $L^2\left(\mathbb{R},\pi^{-\frac{1}{2}} e^{-x^2}dx\right)$.
\end{theorem}
\begin{proof}
The kernel function \eqref{Lauricella} can be obtained by putting $m=0$ in \eqref{B_beta,m} of Theorem 4.1.1. To deduce \eqref{B_00} one has to use the fact that
\begin{eqnarray}
F_{1:0;0;0}^{1:0;0;1}\left(
\begin{array}{c}
[1:1,2,1]:-;-;[0 :1] \\ \left[1:1,2,2\right]:-;-;-
\end{array}
 \sqrt{2}x\overline{z}, -\overline{z}^2/2, -\overline{z}^2\right)=e^{-\frac{1}{2}\bar{z}^2+\sqrt{2}x\bar{z}}
\end{eqnarray}
which can be checked by direct calculation.
\end{proof}
\begin{proposition}
For $\beta \geq 0$, the operator $L_1^\beta$ in $(\ref{landau-ops})$ can be expressed in a differential form as 
\begin{eqnarray}
L_1^\beta=-\frac{\partial ^{2}}{\partial z\partial \overline{z}}+\overline{z}\frac{\partial }{\partial \overline{z}} -\frac{\beta }{z}\frac{\partial }{\partial \overline{z}}+\beta +\frac{1}{2}.\label{L_2_beta}
\end{eqnarray}
\end{proposition}
\begin{proof}
Here the obtained class of 2D orthogonal polynomials obtained are $H_{n,m}^{(\beta)}(z,\overline{z})$ given in $(\ref{GHPoly})$. According to (\cite{MZ}, Theorem 3.4), these polynomials satisfy the following second order partial differential equation
\begin{eqnarray}
\left(\frac{\beta}{z}+\frac{\partial}{\partial z}-\overline{z}\right) \left(-\frac{\partial}{\partial \overline{z}}\right)H_{n,m}^{(\beta)}(z,\overline{z})&=&m H_{n,m}^{(\beta)}(z,\overline{z}),\qquad n\geq m \label{a_2^dag a_2}
\end{eqnarray}
By comparing the actions of the operator in the left hand side of $(\ref{a_2^dag a_2})$ , on the basis vectors $\widetilde{P}_{n,m}^{\beta}=H_{n,m}^{(\beta)}(z,\overline{z})$, with the actions of operator $L_1^{\beta}$ as in $(\ref{landau-ops})$, respectively, on the same basis, we deduce the expressions $(\ref{L_2_beta})$.
\end{proof}
As we will see below the range of the coherent states transform also admits a realization as the null space of a generalized Landau operator.

\begin{proposition}
Let $\beta \geq 0$. Then, the subspace $\mathcal{A}_{\beta}^2(\mathbb{C})$ turns out to be the null space of the generalized Landau Hamiltonian
\begin{equation}
\tilde{\Delta}_\beta :=-\frac{\partial ^{2}}{\partial z\partial \overline{z}}+\overline{z}\frac{\partial }{\partial \overline{z}} -\frac{\beta }{z}\frac{\partial }{\partial \overline{z}}=L_1^\beta -\beta -\frac{1}{2}.
\end{equation}
That is
\begin{equation}
\mathcal{A}_{\beta}^2(\mathbb{C})=\left\{\phi \in L^{2,\beta}(\mathbb{C}),\quad \tilde{\Delta}_\beta\left[ \phi \right] =0 \right\}.
\end{equation}
\end{proposition}
\begin{proof}
Denote $\mathcal{E}_{0}^{\beta }:=\left\{ \phi \in L^{2,\beta}(\mathbb{C}),\quad \tilde{\Delta}_\beta\left[ \phi \right] =0\right\}.$ Let $\phi \in \mathcal{A}_{\beta}^2(\mathbb{C}) .$ So $\phi \in 
L^{2,\beta}(\mathbb{C})$ and $\phi $ is entire. Then
\begin{equation*}
\frac{\partial }{\partial \overline{z}}\left[ \phi \right] =0
\end{equation*}%
We apply to this equation the operator%
\begin{equation*}
\left( \frac{\partial }{\partial \overline{z}}\right) ^{\ast }=-\frac{%
\partial }{\partial z}-\frac{\beta }{z}+\overline{z}
\end{equation*}%
so we still have 
\begin{equation*}
\left( \frac{\partial }{\partial \overline{z}}\right) ^{\ast }\frac{\partial 
}{\partial \overline{z}}\left[ \phi \right] =0.
\end{equation*}%
This means%
\begin{equation*}
\tilde{\Delta}_\beta\left[ \phi \right] =0.
\end{equation*}%
Therefore, $\phi \in \mathcal{E}_{0}^{\beta }$. We have proved that $\mathcal{A}_{\beta}^2(\mathbb{C}) \subset \mathcal{E}_{0}^{\beta
}.$ Conversely, let $\varphi \in \mathcal{E}_{0}^{\beta }.$ Then, $\tilde{\Delta}_\beta\left[ \varphi \right] =0$ which means that%
\begin{equation*}
\left\langle \tilde{\Delta}_\beta\left[ \varphi \right] ,\varphi
\right\rangle =0.
\end{equation*}%
By (\cite{MZ}, Theorem 3.5.) the operator $\tilde{\Delta}_\beta$ is positive in the sense that $\left\langle \tilde{\Delta}_\beta\left[ \varphi \right],\varphi
\right\rangle  \geq 0$ with 
\begin{equation}
\left\langle \tilde{\Delta}_\beta\left[ \varphi \right],\varphi
\right\rangle =0\text{ \ \ \ if and only if \ \ \ }\frac{\partial }{\partial 
\overline{z}}\varphi =0\text{.} 
\end{equation}%
This implies that%
\begin{equation*}
\text{\ \ }\frac{\partial }{\partial \overline{z}}\varphi =0
\end{equation*}%
which means that $\varphi $ is an entire function. That is $\mathcal{E}%
_{0}^{\beta }\subset \mathcal{A}_{\beta}^2(\mathbb{C}).$
\end{proof}
\begin{remark}
Note that for $\beta =0$ the operator $(\ref{L_2_beta})$ reduces to 
\begin{equation}
\widetilde{\Delta }_0:=-\frac{\partial ^{2}}{\partial z\partial \overline{z}}+\overline{z}\frac{\partial }{\partial \overline{z}} =L_1^0-\frac{1}{2} \label{delta0}
\end{equation}
which can be obtained by intertwining (unitarily) the Hamiltonian describing the
dynamics of a charged particle on the Euclidean $xy$-plane, while
interacting with a perpendicular constant homogeneous magnetic field (in
suitable unit system): 
\begin{equation}
H^{L}:=\frac{1}{2}\left( \left( i\frac{\partial }{\partial x}-y\right)
^{2}+\left( i\frac{\partial }{\partial y}+x\right) ^{2}\right) 
\label{2.1.3}
\end{equation}
acting on the Hilbert space $L^{2}\left( \mathbb{R}^{2},dxdy\right) $ as
follows 
\begin{equation}
e^{\frac{1}{2}z\overline{z}}\left( \frac{1}{2}H^{L}-\frac{1}{2}\right) e^{-%
\frac{1}{2}z\overline{z}}=\widetilde{\Delta }_0 ,\quad z=x+iy. 
\end{equation}
The spectrum of the Hamiltonian $\frac{1}{2}H^{L}$ consists on $E_{n}:=n+%
\frac{1}{2}$, $n=0,1,2,....$ known as Euclidean Landau levels with infinite
degeneracy. The operator $\widetilde{\Delta }_0$ is acting on the Hilbert
space $L^{2,0}(\mathbb{C}):=L^{2}\left( \mathbb{C},e^{-z\overline{z}}d\nu
 \right) $ of Gaussian square integrable functions. Its spectrum in $L^{2,0}(\mathbb{C})$ consists on $\epsilon _{n}:=n,$ $%
n=0,1,2,...$ .
\end{remark}
\begin{appendix}
\section{The proof of proposition 3.2.1.}
Assuming the finiteness of the sum
\begin{equation}
\mathcal{N}_{\beta ,m}(z\overline{z} )=\sum_{n=0}^{\infty } \frac{\vert P_{n,m}^\beta(z,\overline{z})\vert^2}{x_{n,m}^{\beta}!}
<+\infty ,  \label{Norma}   
\end{equation}
means that 
\begin{equation}
\sum_{n=0}^{m-1 } \frac{\vert P_{n,m}^\beta(z,\overline{z})\vert^2}{x_{n,m}^{\beta}!}+\sum_{n=m}^{\infty } \frac{\vert P_{n,m}^\beta(z,\overline{z})\vert^2}{x_{n,m}^{\beta}!}
<+\infty \label{sum+serie}
\end{equation}
this implies the finiteness of the infinite sum in $(\ref{sum+serie})$ which may be written as
\begin{eqnarray*}
\sum_{n=m}^{\infty }  \frac{|\phi _{m}(z\overline{z};n-m+\beta
)|^2}{\zeta _{n\wedge m}\left( \left\vert n-m\right\vert +\beta \right)}(z\overline{z})^{n-m}
&=& \sum_{n=m}^{\infty } \left(\sum_{i,j=0}^{m}\frac{c_i(m;n-m+\beta)c_j(m;n-m+\beta)}{\zeta _{n\wedge m}\left( \left\vert n-m\right\vert +\beta \right)}(z\overline{z})^{2m-i-j} \right)(z\overline{z})^{n-m}\\
&=& \sum_{n=0}^{\infty } \left(\sum_{i,j=0}^{m}\frac{c_i(m;n+\beta)c_j(m;n+\beta)}{\zeta _{m}( n +\beta )}(z\overline{z})^{2m-i-j} \right)(z\overline{z})^{n}\\
&=& \sum_{i,j=0}^{m} (z\overline{z})^{2m-i-j}  \left(\sum_{n=0}^{\infty } \frac{c_i(m;n+\beta)c_j(m;n+\beta)}{\zeta _{m}( n +\beta )} (z\overline{z})^{n}\right).
\end{eqnarray*}
The radius of convergence of the series 
\begin{equation}
\sum_{n=0}^{\infty } \frac{c_i(m;n+\beta)c_j(m;n+\beta)}{\zeta _{m}( n +\beta )} (z\overline{z})^{n}
\end{equation}
can be found by applying the ratio test and it is given by $(\ref{R_beta,n,i,j})$. $\Box$

\section{The proof of proposition 4.1.1.}
Let
\begin{eqnarray*}
S &=& \sum\limits_{j=0}^{+\infty }\frac{%
\left( j\wedge m\right) !}{\Gamma \left( \beta +j\vee m+1\right) }%
H_{j,m}^{\left( \beta \right) }\left( z,\overline{z}\right) \overline{H_{j,m}^{\left(
\beta \right) }\left( w,\overline{w}\right)}\\
&=&\sum\limits_{j=0}^{m-1 }\frac{ j !}{\Gamma \left( \beta + m+1\right) }H_{j,m}^{\left( \beta \right) }\left( z,\overline{z}\right) \overline{H_{j,m}^{\left(
\beta \right) }\left( w,\overline{w}\right)} +\sum\limits_{j=m}^{+\infty }\frac{
 m !}{\Gamma \left( \beta +j+1\right) }H_{j,m}^{\left( \beta \right) }\left( z,\overline{z}\right) \overline{H_{j,m}^{\left(
\beta \right) }\left( w,\overline{w}\right)}\\ 
&=&\sum\limits_{j=0}^{m-1 }\frac{ j !(\overline{z}w)^{m-j}}{\Gamma \left( \beta + m+1\right) }
L_{j}^{\left( \beta +m-j\right) }\left( z\overline{z}\right) L_{j}^{\left( \beta +m-j\right) }\left( w\overline{w}\right)+\sum\limits_{j=m}^{+\infty }\frac{m !}{\Gamma \left( \beta +j+1\right) }H_{j,m}^{\left( \beta \right) }\left( z,\overline{z}\right) \overline{H_{j,m}^{\left(
\beta \right) }\left( w,\overline{w}\right)}.
\end{eqnarray*}
Now, to obtain a closed form of the infinite series in the last equation 
\begin{equation}
S_{(\infty)}=\sum\limits_{j=m}^{+\infty }\frac{m !}{\Gamma \left( \beta +j+1\right) }H_{j,m}^{\left( \beta \right) }\left( z,\overline{z}\right) \overline{H_{j,m}^{\left(
\beta \right) }\left( w,\overline{w}\right)}
\end{equation}
we use $(\ref{GHPoly})$, after substitution and some simplifications, we have
\begin{eqnarray}
S_{(\infty)}=\frac{1}{m!\Gamma(\beta +1)}\sum\limits_{j=m}^{+\infty }(\beta +1)_{j}\sum\limits_{k=0}^{m}\sum\limits_{l=0}^{m} \left( 
\begin{array}{c}
m \\ 
k
\end{array}
\right)\left( 
\begin{array}{c}
m \\ 
l
\end{array}
\right)
\frac{(-1)^{k+l}z^{j-k}\overline{z}^{m-k}\overline{w}^{j-l}w^{m-l}}{(\beta +1)_{j-k}(\beta +1)_{j-l}}
\end{eqnarray}
By changing $j\rightarrow j-m$ and the summation order, it follows
\begin{eqnarray*}
S_{(\infty)}=\frac{1}{m!\Gamma(\beta +1)}\sum\limits_{k=0}^{m}\sum\limits_{l=0}^{m} \left( 
\begin{array}{c}
m \\ 
k
\end{array}
\right)\left( 
\begin{array}{c}
m \\ 
l
\end{array}
\right)(-1)^{k+l}
\sum\limits_{j=0}^{+\infty }(\beta +1)_{j+m}\frac{z^{j+m-k}\overline{z}^{m-k}\overline{w}^{j+m-l}w^{m-l}}{(\beta +1)_{j+m-k}(\beta +1)_{j+m-l}}
\end{eqnarray*}
which also can be written as
\begin{eqnarray}
S_{(\infty)}=\frac{1}{m!\Gamma(\beta +1)}\sum\limits_{k=0}^{m}\sum\limits_{l=0}^{m} \left( 
\begin{array}{c}
m \\ 
k
\end{array}
\right)\left( 
\begin{array}{c}
m \\ 
l
\end{array}
\right)(-1)^{k+l}
\sum\limits_{j=0}^{+\infty }(\beta +1)_{j+m}\frac{z^{j+k}\overline{z}^{k}\overline{w}^{j+l}w^{l}}{(\beta +1)_{j+k}(\beta +1)_{j+l}}.
\end{eqnarray}
By the following proprieties
\begin{eqnarray*}
(a)_{k+m}=(a+m)_k(a)_m \quad\text{and} \quad (-m)_k=(-1)^k\frac{m!}{(m-k)!}=(-1)^kk!\left( 
\begin{array}{c}
m \\ 
k
\end{array}
\right),
\end{eqnarray*}
we arrive at
\begin{eqnarray*}
S_{(\infty)}= \frac{(\beta+1)_m}{m!\Gamma(\beta+1)}\sum\limits_{k=0}^{m}\sum\limits_{l=0}^{m}\frac{(-m)_k(-m)_l(z\overline{z})^k(w\overline{w})^{l}}{k!l!(\beta+1)_k(\beta+1)_{l}}{}_{2}F_{2}\left( 
\begin{array}{c}
1,m+\beta +1 \\ 
k+\beta +1,l+\beta +1
\end{array}%
\big|z\overline{w}\right).
\end{eqnarray*}
\section{Generalized Lauricella series}
The generalized Lauricella series in several variables is defined by (\cite{Sriva},p.36):
\begin{eqnarray}
&&F_{C:D^{(1)};\cdots ;D^{(n)}}^{A:B^{(1)};\cdots ;B^{(n)}}\left[ 
\begin{array}{c}
[(a):\theta^{(1)},\cdots ,\theta^{(n)}]:[(b^{(1)}):\phi^{(1)}];\cdots [(b^{(n)}):\phi^{(n)}] \\ 

[(c):\psi^{(1)},\cdots ,\psi^{(n)}]:[(b^{(1)}):\delta^{(1)}];\cdots [(b^{(n)}):\delta^{(n)}]
\end{array} z_1,...,z_n
\right]\\
&& :=\sum\limits_{m_1,...,m_n=0}^{+\infty}\Omega(m_1,...,m_n)\frac{z_1^{m_1}}{m_1!}...\frac{z_n^{m_n}}{m_n!}
\end{eqnarray}
where, for convenience,
\begin{eqnarray}
\Omega(m_1,...,m_n):=\frac{\prod\limits_{j=1}^{A}(a_j)_{m_1\theta_j^{(1)}+...+m_n\theta_j^{(n)}}\prod\limits_{j=1}^{B^{(1)}}(b_j^{(1)})_{m_1\phi_1^{(1)}}\cdots \prod\limits_{j=1}^{B^{(n)}}(b_j^{(n)})_{m_n\phi_j^{(n)}}}{\prod\limits_{j=1}^{C}(c_j)_{m_1\psi_j^{(1)}+...+m_n\psi_j^{(n)}}\prod\limits_{j=1}^{D^{(1)}}(d_j^{(1)})_{m_1\delta_1^{(1)}}\cdots \prod\limits_{j=1}^{D^{(n)}}(d_j^{(n)})_{m_n\delta_j^{(n)}}},
\end{eqnarray}
the coefficients
\begin{eqnarray}
\left\{ 
\begin{array}{c}
\theta_j^{(k)},\; j=1,...,A,\quad \phi_j^{(k)},\; j=1,...,B^{(k)},\\ 
\psi_j^{(k)},\; j=1,...,C, \quad \delta_j^{(k)},\; j=1,...,D^{(k)},
\end{array} k=1,...,n
\right.
\end{eqnarray}
are real and positive, and $(a)$ abbreviates the array of $A$ parameters $a_1,...,a_A$; $(b^{(k)})$ abbreviates the array of $(B^{(k)})$ parameters 
$b_j^{(k)}, \; j=1,...,B^{(k)},\; k=1,...,n$.
 Similar interpretation hold for the remaining parameters. For precise conditions under which the generalized Lauricella function converges, see Srivastava and Daoust (1972, pp. 153-157), also see Exton (1976, Sect. 3.7) and Exton (1978, Sect. 1.4).
\section{The wave functions of coherent states}
The wave functions of coherent states (\ref{GNLCS}) are given by
\begin{equation}
\vartheta_{z,m,\beta}(x) = \left(\mathcal{N}_{\beta ,m} (z\overline{z})\;\right)^{-\frac 12}\;B_{\beta,m}(x,z)
\end{equation}
where
\begin{equation}
 B_{\beta,m}(x,z)= \sum_{j=0}^{+\infty} \sqrt{\frac{%
\left( j\wedge m\right) !}{\Gamma \left( \beta +j\vee m+1\right) }}
\frac{2^{j/2}}{\sqrt{(\beta +1)_j}}\overline{H_{j,m}^{\left(\beta \right) }\left( z,\overline{z}\right)}H_j(x,\beta ).
\end{equation}
We now express the polynomials $H_{j,m}^{\left(\beta \right) }\left( z,\overline{z}\right)$ in terms of Laguerre polynomials as in $(\ref{P_mn for all m,n})$: 
\begin{equation}
H_{j,m}^{\left(\beta \right) }\left( z,\overline{z}\right)=(-1)^{j \wedge m}z^{j-(j \wedge m)}\overline{z}^{m-(j \wedge m)}L^{(|j-m|+\beta)}_{j \wedge m}(z\overline{z}
), \quad j,m\in \mathbb{N},
\end{equation}
and we split the sum $\sum_{j=0}^{+\infty}$ in two sums $S_{m-1}=\sum_{j=0}^{m-1}$ and $S_{(\infty)}=\sum_{j=m}^{+\infty}$ as
\begin{eqnarray*}
B_{\beta,m}(x,z)&=& \sum_{j=0}^{m-1} (-1)^{j}\sqrt{\frac{j !}{\Gamma \left( \beta + m+1\right) }}
\frac{2^{j/2}z^{m-j}}{\sqrt{(\beta +1)_j}}L^{(m-j+\beta)}_{j}(z\overline{z}
)H_j(x,\beta )\\&+&(-1)^{m}\sum_{j=m}^{+\infty} \sqrt{\frac{m !}{\Gamma \left( \beta + j+1\right) }}
\frac{2^{j/2}\bar{z}^{j-m}}{\sqrt{(\beta +1)_j}}L^{(j-m+\beta)}_{m}(z\overline{z}
)H_j(x,\beta )=S_{m-1}+S_{(\infty)}.
\end{eqnarray*}
Again, we write the infinite sum in the above equation as 
\begin{eqnarray*}
S_{(\infty )}&=&(-1)^m\sum_{j=0}^{+\infty} \sqrt{\frac{m!}{\Gamma \left( \beta + j+1\right) }}
\frac{2^{j/2}\bar{z}^{j-m}}{\sqrt{(\beta +1)_j}}L^{(j-m+\beta)}_{m}(z\overline{z}
)H_j(x,\beta )\\
&-& (-1)^m\sum_{j=0}^{m-1} \sqrt{\frac{m !}{\Gamma \left( \beta + j+1\right) }}
\frac{2^{j/2}\bar{z}^{j-m}}{\sqrt{(\beta +1)_j}}L^{(j-m+\beta)}_{m}(z\overline{z}
)H_j(x,\beta )=S_{(\infty)}^*-S_{m-1}^*
\end{eqnarray*}
To compute $S_{(\infty)}^*$, we use the expression of the Laguerre polynomial 
\begin{equation}
L^{(j-m+\beta)}_{m}(z\overline{z}
)=\frac{1}{m!}\sum_{k=0}^{m}\left( 
\begin{array}{c}
m \\ 
k
\end{array}%
\right) (j+\beta +1-k)_k
\left( -z\bar{z}\right)^{m-k}, 
\end{equation}
together with the relation
\begin{eqnarray}
(j+\beta +1-k)_k=\frac{(\beta +1-k)_{k}(\beta +1)_{j}}{(\beta +1-k)_{j}}
\end{eqnarray}
in order to rewrite the Laguerre polynomial as
\begin{equation}
L^{(j-m+\beta)}_{m}(z\overline{z}
)=\frac{1}{m!}\sum_{k=0}^{m}\left( 
\begin{array}{c}
m \\ 
k
\end{array}%
\right)\frac{(\beta +1-k)_{k}(\beta +1)_{j}}{(\beta +1-k)_{j}}\left( -z\bar{z}\right)^{m-k}. \label{Laguerre(j-m+beta)}
\end{equation}
We substitute \eqref{Laguerre(j-m+beta)} in the expression of $S_{\infty}^*$ and we change the summation order to get
\begin{eqnarray}
S_{(\infty)}^* =\frac{z^m}{\sqrt{\Gamma \left( \beta +1\right)m!}}\sum_{k=0}^{m}\left( 
\begin{array}{c}
m \\ 
k
\end{array}%
\right)\frac{(-1)^k(\beta +1-k)_{k}}{\left(z\bar{z}\right)^{k}}\sum_{j=0}^{+\infty} \frac{
(\bar{z}/\sqrt{2})^{j}}{(\beta +1-k)_{j}}H_j(x,\beta )\label{S_inf}
\end{eqnarray}
To obtain a closed form for the infinite sum in $(\ref{S_inf})$ we make use of the following lemma.
\begin{lemma}
A generating functions for the associated Hermite polynomials is given by
\begin{eqnarray}
\sum_{n=0}^\infty \frac{t^n}{(c)_n}H_n(x,\beta )=F_{1:0;0;0}^{1:0;0;1}\left(
\begin{array}{c}
[1:1,2,1]:-;-;[\beta :1] \\ \left[c:1,2,2\right]:-;-;-
\end{array}
 2xt, -t^2, -2t^2\right).\label{GFAHP}
\end{eqnarray}
in terms of generalized Lauricella functions $F_{1:0;0;0}^{1:0;0;1}$. In particular, for $c=1$ and $\beta =0$, \eqref{GFAHP} reduces to
\begin{equation}
\sum_{n=0}^\infty \frac{t^n}{n!}H_n(x)=e^{2xt-t^2}
\end{equation}
which is the generating function of Hermite polynomials.
\end{lemma}
We apply lemma 1 for parameters $t=\frac{\overline{z}}{\sqrt{2}}$ and $c=\beta -k+1$ to get
\begin{eqnarray}
\sum_{n=0}^\infty \frac{\left(\frac{\overline{z}}{\sqrt{2}}\right)^n}{(\beta -k+1)_n}H_n(x,\beta )=F_{1:0;0;0}^{1:0;0;1}\left(
\begin{array}{c}
[1:1,2,1]:-;-;[\beta :1] \\ \left[\beta -k+1:1,2,2\right]:-;-;-
\end{array}
 \sqrt{2}x\overline{z}, -\frac{\overline{z}^2}{2}, -\overline{z}^2\right).
\end{eqnarray}
Summarizing the above calculations by writing $B_{\beta,m}(z,x)=S_{m-1}-S_{m-1}^*+S_{\infty}^{*}$, we get 
\begin{eqnarray*}
B_{\beta,m}(z,x)=\sum\limits_{n=0}^{m-1}\frac{2^{n/2}}{\sqrt{(\beta +1)_n}}H_n(x,\beta)\left[\frac{(-1)^n z^{m-n}\sqrt{n!}}{\sqrt{\Gamma(\beta +m+1)}}L_n^{(m-n+\beta)}(z\bar{z})-\frac{(-1)^m \bar{z}^{n-m}\sqrt{m!}}{\sqrt{\Gamma(\beta +n+1)}}L_m^{(n-m+\beta)}(z\bar{z}) \right] \\ +\frac{z^m}{\sqrt{\Gamma \left( \beta +1\right)m!}}\sum_{k=0}^{m}\frac{(-m)_k(\beta +1-k)_{k}}{k!\left(z\bar{z}\right)^{k}}F_{1:0;0;0}^{1:0;0;1}\left(
\begin{array}{c}
[1:1,2,1]:-;-;[\beta :1] \\ \left[\beta -k+1:1,2,2\right]:-;-;-
\end{array}
 \sqrt{2}x\overline{z}, -\overline{z}^2/2, -\overline{z}^2\right).
\end{eqnarray*}
\textbf{\textit{Proof of lemma 1.}}
We first write the associated Hermite polynomials $H_n(x,\beta )$ as follows
\begin{equation}
H_n(x,\beta)=\sum\limits_{k=0}^{\lfloor n/2\rfloor}\frac{(-1)^k n!}{k! (n-2k)!}\left(\sum\limits_{j=0}^{k}\frac{(-k)_j (\beta)_j}{(-n)_j}\frac{2^j}{j!}\right)(2x)^{n-2k}
\end{equation}
which are obtained from the polynomials $He_{n}^{\beta}(x)$ in (\cite{Wunsche}, p.203) by the relation $H_n(x,\beta)=2^{\frac{1}{2}n}He_{n}^{\beta}(\sqrt{2}x)$.\\
Similar manipulations as in (\cite{Greubel}, pp.548-549), yield
\begin{eqnarray*}
\sum_{n=0}^\infty \frac{t^n}{(c)_n}H_n(x,\beta )&=& \sum_{n=0}^\infty \sum\limits_{k=0}^{\lfloor n/2\rfloor}\left(\sum\limits_{j=0}^{k}\frac{(-k)_j (\beta)_j}{(-n)_j}\frac{2^j}{j!}\right)\frac{(-1)^k n!}{k! (n-2k)!}(2x)^{n-2k}\frac{t^n}{(c)_n}\\
&=& \sum_{n,k=0}^\infty \sum\limits_{j=0}^{k}\frac{(-k)_j (\beta)_j}{(-n-2k)_j}\frac{2^j}{j!}\frac{(-1)^k (n+2k)!}{k! n!}(2x)^{n}\frac{t^{n+2k}}{(c)_{n+2k}}\\
&=& \sum_{n,k,j=0}^\infty \frac{(n+2k+2j)!(-k-j)_j(1)_k(\beta)_j}{(k+j)!(-n-2k-2j)_j(c)_{n+2k+2j}}\frac{(2xt)^n}{n!}\frac{(-t^2)^k}{k!}\frac{(-2t^2)^j}{j!}
\end{eqnarray*} 
where the second and the last line in the above calculation are, respectively, justified by applying the following formulas (\cite{Sr-Ma}, lemma 2, p.101):
\begin{equation}
\sum_{n=0}^\infty \sum\limits_{k=0}^{\lfloor n/2\rfloor}A(k,n)=\sum_{n=0}^\infty \sum\limits_{k=0}^{\infty}A(k,n+2k),
\end{equation}
and (\cite{Sr-Ma}, lemma 1, p.100):
\begin{equation}
\sum_{k=0}^\infty \sum\limits_{j=0}^{k}A(j,k)=\sum_{k=0}^\infty \sum\limits_{j=0}^{\infty}A(j,k+j).
\end{equation}
But since
\begin{equation}
(1)_{n+2k+j}=\frac{(n+2k+2j)!(-k-j)_j(1)_k}{(k+j)!(-n-2k-2j)_j}
\end{equation}
then we obtain the following equality
\begin{eqnarray}
\sum_{n=0}^\infty \frac{t^n}{(c)_n}H_n(x,\beta )=\sum_{n,k,j=0}^\infty \frac{(1)_{n+2k+j}(\beta)_j}{(c)_{n+2k+2j}}\frac{(2xt)^n}{n!}\frac{(-t^2)^k}{k!}\frac{(-2t^2)^j}{j!}\label{sum1}
\end{eqnarray}
where the right hand side is recognized as the generalized Lauricella function (\cite{Sriva},p.36):
\begin{eqnarray}
F_{1:0;0;0}^{1:0;0;1}\left(
\begin{array}{c}
[1:1,2,1]:-;-;[\beta :1] \\ \left[c:1,2,2\right]:-;-;-
\end{array}
 2xt, -t^2, -2t^2\right).
\end{eqnarray}
This ends the proof of lemma 1. $\Box$
\section{The proof of proposition 4.2.1.}
Let us assume that the measure takes the form $d\eta_{\beta}(z)=\mathcal{N}_{\beta}(z\bar{z})h(z\bar{z})d\mu(z),$
where $h$ is an auxiliary density function to be determined. In terms of
polar coordinates $z=\rho e^{i\theta} ,\ \rho>0$ and $\theta\in [0,2\pi)$, then the measure can be rewritten as
\begin{equation}
\label{Mesure}
d\eta_{\beta}(z)=\mathcal{N}_{\beta}(\rho
^{2})h(\rho ^{2})\rho d\rho \frac{d\theta }{2\pi }.
\end{equation}
Using the expression $(\ref{CSs})$ of coherent states, the operator $\mathcal{O}_{\beta}=\int_{\mathbb{C}}\left| \vartheta_{z,\beta}
\right\rangle \left\langle\vartheta_{z,\beta}\right| d\eta_{\beta}(z)$ reads successively,
\begin{eqnarray}
\mathcal{O}_{\beta}&=&(\Gamma(\beta +1))^{-1}\sum\limits_{n,m=0}^{+\infty} \left(  \int_0^{+\infty} \frac{\rho^{n+m}h(\rho^2)\rho d\rho}{\sqrt{x_n^\beta !}\sqrt{x_m^\beta !}}\left( \int_0^{2\pi} e^{i(n-m)\theta}\frac{d\theta}{2\pi} \right)\right)  \vert \varphi_n\rangle \langle \varphi_m\vert\\
&=&(\Gamma(\beta +1))^{-1}\sum\limits_{n=0}^{+\infty}\frac{1}{(\beta +1)_n}\left(  \int_0^{+\infty}\rho^{2n}h(\rho^2) \rho d\rho \right)  \vert \varphi_n\rangle^\beta \langle \varphi_n^\beta\vert\\ 
&=& \sum\limits_{n=0}^{+\infty}\frac{1}{2\Gamma(n+\beta +1)}\left(  \int_0^{+\infty}r^{n}h(r)dr \right)  \vert \varphi_n\rangle^\beta \langle \varphi_n^\beta\vert.\label{4.14}
\end{eqnarray}
Now, to determinate $h$ such that
\begin{equation}
\label{Meijer}
\int_0^{+\infty}r^{n}h(r)dr=2\Gamma(n+\beta +1),
\end{equation}
we recall the Euler gamma function
\begin{equation}
\int_0^{+\infty}t^{s}e^{-t}dt=\Gamma(s+1), \quad Re(s)>-1.
\end{equation}
This suggests us to take the weight function $h(r)=2 r^{\beta}e^{-r}$. This leads to the measure in $(\ref{measure-RI})$. With this measure, Eq.$(\ref{4.14})$ reduces to $\mathcal{O}_{\beta}=\sum_{n=0}^{+\infty}\vert \varphi_n^\beta\rangle \langle \varphi_n^\beta\vert=\textbf{1}_{\mathcal{H}},$ since $\{\varphi_{n}^\beta\}$ is an orthonormal basis of $L^2\left(\mathbb{R},\; d\omega_{\beta}(x)\right)$. $\Box$

\end{appendix}


\end{document}